%% file: main.tex
\pdfoutput=1
\documentclass[conference]{IEEEtran} 

\usepackage[english]{babel}

\usepackage{microtype}

\usepackage[cmex10]{amsmath}
\usepackage{mathrsfs,amssymb}
\usepackage{style/mathpartir}
\usepackage[numbers]{natbib}

\input{style/layout}
\input{style/notation}

\input{style/theorems}

\let\amsproof\proof
\let\endamsproof\endproof

\usepackage{graphicx}
\def\tikzset#1{}%
\def\tikzsetnextfilename#1{}%
\def\inputfig figures/#1.tikz{\includegraphics{export/#1}}
\newenvironment{tikzpicture}[1][]{%
  \def\input##1{\inputfig##1}%
}{}

\input{style/shortversion}

\shortversiontrue
\shortversionfalse 
\includeappendixfalse 

\begin{document}
\let\proof\amsproof
\let\endproof\endamsproof

\title{A Type System for proving Depth Boundedness in the \picalc}

\author{%
  \IEEEauthorblockN{Emanuele D'Osualdo}%
  \IEEEauthorblockA{%
    University of Oxford, UK\\
    emanuele.dosualdo@cs.ox.ac.uk}
  \and
  \IEEEauthorblockN{Luke Ong}%
  \IEEEauthorblockA{
    University of Oxford, UK\\
    lo@cs.ox.ac.uk}
  }

\maketitle

\begin{abstract}
The depth-bounded fragment of the \picalc\ is an expressive class of systems enjoying decidability of some important verification problems.
Unfortunately membership of the fragment is undecidable.
We propose a novel type system, parameterised over a finite forest,
that formalises name usage by $\pi$-terms in a manner that respects the forest.
Type checking is decidable and type inference is computable;
furthermore typable \piterm{s} are guaranteed to be depth bounded.

The second contribution of the paper is a proof of equivalence between the semantics of typable terms and nested data class memory automata,
a class of automata over data words.
We believe this connection can help to establish new links
between the rich theory of infinite-alphabet automata and nominal calculi.
\end{abstract}

\input{1-intro}

\input{2-picalc}

\input{3-tcompat}

\input{4-typesys}

\input{5-encodings}

\input{6-discussion}

\section*{Acknowledgement}
We would like to thank Damien Zufferey for helpful discussions
on the nature of depth boundedness.

\bibliographystyle{abbrvnat}
\bibliography{biblio}

\ifshortversion
\ifincludeappendix

  \clearpage

  \appendices

  \input{appendix/2-picalc}
  \input{appendix/3-tcompat}
  \input{appendix/4-typesys}
  \input{appendix/5-encodings}
  \input{appendix/6-discussion}

\fi
\fi

\end{document}

%% file: style/notation.tex
\usepackage{centernot}
\usepackage[Symbol]{upgreek}
\usepackage{stmaryrd}

\usepackage{mathtools}

\let\originalleft\left
\let\originalright\right
\renewcommand{\left}{\mathopen{}\mathclose\bgroup\originalleft}
\renewcommand{\right}{\aftergroup\egroup\originalright}

\makeatletter

\let\ams@cases@copy@\cases
\def\cases{%
  \newcommand{\CASE}{&\text{if }}%
  \newcommand{\AND}{\\&\text{and }}%
  \newcommand{\OTHERWISE}{& \text{otherwise}}%
  \ams@cases@copy@%
}

\let\stdphi\phi
\let\phi\varphi
\let\varphi\stdphi

\def\pre#1-{$#1$\nobreak-\nobreak\hskip0pt}

\newcommand{\picalc}{%
    \texorpdfstring{\pre\pi-calculus}{pi-calculus}}

\newcommand{\piterm}{%
    \texorpdfstring{\pre\pi-term}{pi-term}}

\newcommand{\PiTerms}{\mathcal{P}}
\newcommand{\PiNf}{\PiTerms_{\!\mathsf{nf}}}
\newcommand{\PiAnnot}[1][\Types]{\PiTerms^{#1}}
\newcommand{\PiNfAnnot}[1][\Types]{\PiNf^{#1}}
\newcommand{\Seq}{\mathcal{S}}

\newcommand{\Forests}{\mathcal{F}}

\newcommand{\tcompat}[1][ible]{%
  \texorpdfstring{\pre\Types-compat#1}{T-compat#1}}
\newcommand{\tshaped}[1][]{%
  \texorpdfstring{\pre\Types-shaped#1}{T-shaped#1}}

\newcommand{\bigO}[1]{\ensuremath{O(#1)}}

\let\arrowvect\vec

\newcommand{\vect}[1]{\arrowvect{#1}\@ifnextchar]{\,}{\@ifnextchar){\,}{\@ifnextchar{\rangle}{\,}{}}}}
\renewcommand{\vec}[1]{\mathbf{#1}}

\newcommand{\sem}[1]{\llbracket#1\rrbracket}

\newcommand{\Nat}{\mathbb{N}}

\newcommand{\union}{\cup}
\newcommand{\inters}{\cap}
\newcommand{\Union}{\bigcup}

\newcommand{\dunion}{\uplus}
\newcommand{\Dunion}{\biguplus}
\DeclarePairedDelimiter\card{\lvert}{\rvert}

\providecommand{\implies}{\Rightarrow}

\newcommand{\domain}{\operatorname{dom}}

\newcommand{\lst}[3][1]{{#2_{#1}}\ldots {#2_{#3}}}    %
\newcommand{\lstc}[3][1]{{#2_{#1}},\ldots, {#2_{#3}}} %

\providecommand{\coloneq}{\mathrel{\mathop:}=} 
\providecommand{\Coloneqq}{\mathrel{\mathop{::}}=} 
\newcommand{\is}{\coloneq}

\def\congr{\equiv}

\newcommand{\reach}{\operatorname{Reach}}

\newcommand{\from}{\colon}
\newcommand{\pto}{\rightharpoonup}
\newcommand{\inv}[1]{#1^{-1}}

\newcommand{\st}{.\:}  

\let\SavedDoubleVert\relax

{\catcode`\|=\active
  \xdef\set{\protect\expandafter\noexpand\csname set \endcsname}
  \expandafter\gdef\csname set \endcsname#1{\begingroup\mathinner%
  \ifx!#1!%
      \emptyset%
  \else%
      {\lbrace\mathcode`\|32768\let|\midvert #1\rbrace}%
  \fi%
  \endgroup%
  }
  \xdef\Set{\protect\expandafter\noexpand\csname Set \endcsname}
  \expandafter\gdef\csname Set \endcsname#1{%
      \ifx!#1!%
         \emptyset%
      \else%
         \left\{%
         \ifx\SavedDoubleVert\relax \let\SavedDoubleVert\|\fi
         \,{\let\|\SetDoubleVert
         \mathcode`\|32768\let|\SetVert
         #1}\,\right\}%
      \fi%
  }
}
\def\midvert{\egroup\mid\bgroup}
\def\SetVert{\@ifnextchar|{\|\@gobble}
    {\egroup\;\mid@vertical\;\bgroup}}
\def\SetDoubleVert{\egroup\;\mid@dblvertical\;\bgroup}

\begingroup
 \edef\@tempa{\meaning\middle}
 \edef\@tempb{\string\middle}
\expandafter \endgroup \ifx\@tempa\@tempb
 \def\mid@vertical{\middle|}
 \def\mid@dblvertical{\middle\SavedDoubleVert}
\else
 \def\mid@vertical{\mskip1mu\vrule\mskip1mu}
 \def\mid@dblvertical{\mskip1mu\vrule\mskip2.5mu\vrule\mskip1mu}
\fi

\newcommand{\map}[1]{%
    \if\relax\noexpand#1\relax%
        \emptyset%
    \else%
        [\,\@map#1;\relax\noexpand\@end@map\,]%
    \fi%
}
\newcommand{\Map}[1]{%
    \if\relax\noexpand#1\relax%
        \emptyset%
    \else%
        \left[\,\@map#1;\relax\noexpand\@end@map\,\right]%
    \fi%
}
\def\@map#1;#2\@end@map{
    \ifx\relax#2\relax%
        \map@binding[#1]
    \else%
        \map@binding[#1],\:\@map#2\@end@map
    \fi%
}
\def\map@binding[#1->#2]{#1\mapsto#2}

\newcommand{\subst}[1]{%
    \if\relax\noexpand#1\relax%
    \else%
        [\,\@subst#1,\relax\noexpand\@end@subst\,]%
    \fi%
}
\def\@subst#1,#2\@end@subst{
    \ifx\relax#2\relax%
        \subst@binding[#1]
    \else%
        \subst@binding[#1],\:\@subst#2\@end@subst
    \fi%
}

\def\subst@binding[#1->#2]{#2/#1}

\def\grammOr{\hspace{3pt}\mid\hspace{3pt}}
\def\grammIs{\Coloneqq}

\begingroup
\catcode`\|=\active%
\gdef\@grammar@bar{%
    \catcode`\|=\active%
    \def|{\grammOr}%
}
\endgroup

\newcommand{\gramm}[1]{%
  \begingroup
  \def\is{\grammIs}%
  \@grammar@bar%
  #1%
  \endgroup%
}

\newenvironment{grammar}{%
    \begin{equation*}%
    \def\is{& \grammIs }%
    \@grammar@bar%
    \aligned%
}
{%
    \endaligned%
    \end{equation*}%
    \aftergroup\ignorespaces%
}

\newcommand{\Names}{\ensuremath{\mathcal{N}}}

\def\out#1<#2>{\overline{#1}\langle#2\rangle}
\def\inp#1(#2){{#1}(#2)}
\def\tact{\boldsymbol\tau} 
\def\outz#1{\overline{#1}}
\def\new#1.{\restr #1.\ignorespaces}

\newcommand{\restr}{\upnu}  
\newcommand{\zero}{\mathbf{0}}
\newcommand{\freenames}{\operatorname{fn}}
\newcommand{\boundnames}{\operatorname{bn}}
\newcommand{\resboundnames}{\operatorname{bn}_\nu}

\newcommand{\actrestr}{\operatorname{active}_{\restr}}
\newcommand{\seqproc}{\operatorname{seq}}

\newcommand{\nf}{\operatorname{nf}}

\newcommand{\nestr}{\operatorname{nest}_{\restr}}
\newcommand{\depth}{\operatorname{depth}}

\newcommand{\emptyforest}{(\emptyset, \emptyset)}
\newcommand{\forest}{\operatorname{forest}}
\newcommand{\AST}{\operatorname{\mathcal{F}}\sem}
\newcommand{\height}{\operatorname{height}}

\newcommand{\minrestr}{\operatorname{min}_{\Types}}

\newif\if@initial@pi@term@

\newcommand{\redto}{\to}

\newcommand{\bang}[1]{#1^{*}}

\newcommand{\Parallel}{\@ifstar{\prod}{{\textstyle\prod}}}
\newcommand{\Alt}{\@ifstar{\sum}{{\textstyle\sum}}}

\newcommand{\linkedto}[1]{\leftrightarrow_{#1}}
\newcommand{\tiedto}[1]{\smallfrown_{#1}}
\newcommand{\ntiedto}[1]{\triangleleft_{#1}}
\newcommand{\migr}[1]{\operatorname{Mig}_{#1}}

\newcommand{\Types}{\ensuremath{\mathcal{T}}}

\newcommand{\Env}{\Gamma}
\newcommand{\types}[1][\Types]{\vdash_{#1}}
\newcommand{\type}{\tau} 
\newcommand{\tas}{\,{:}\,} 
\newcommand{\ty}[1]{\mathsf{#1}} 

\newcommand{\tleq}{\leq}
\newcommand{\tlt}{<}

\newcommand{\parent}{%
  \@ifnextchar+{\p@rentop@plus}{%
  \@ifnextchar*{\p@rentop@star}{%
  \p@rentop}}}
\stmry@if\DeclareMathSymbol\YleftRel\mathrel{stmry}{"06}\fi
\newcommand{\p@rentop}{\YleftRel}
\newcommand{\p@rentop@plus}[1]{<}
\newcommand{\p@rentop@star}[1]{\leq}

\newcommand{\typevar}{\mathfrak{t}}  

\newcommand{\paths}{\operatorname{paths}}
\newcommand{\traces}{\operatorname{traces}}
\newcommand{\nodes}{\operatorname{nodes}}

\newcommand{\Premise}{\Psi}

\newcommand{\phimig}{\phi_{\mathrm{mig}}}
\newcommand{\phinonmig}{\phi_{\neg\mathrm{mig}}}

\newcommand{\treeins}{\operatorname{ins}}

\newcommand{\base}{\operatorname{base}}

\newcommand{\NDA}{NDCMA}  
\newcommand{\dataset}{\ensuremath{\mathcal{D}}}
\newcommand{\CMF}[1][\dataset, \States]{\operatorname{CMF}(#1)}
\newcommand{\Aut}{\ensuremath{\mathcal{A}}}

\newcommand{\AutEnc}{\operatorname{\mathcal{A}}\sem}
\newcommand{\ProcEnc}{\operatorname{\mathcal{P}}\sem}

\newcommand{\pred}{\operatorname{pred}}
\newcommand{\fresh}{\mathfrak{f}}

\newcommand{\deriv}{\operatorname{der}}
\newcommand{\Deriv}[1][P]{\Delta_{#1}}

\newcommand{\States}{\mathbb{Q}}
\newcommand{\q}[1]{q_{\text{#1}}}
\newcommand{\qready}{\q{ready}}
\newcommand{\qdead}{q_\dagger}
\newcommand{\s}[2]{#1^{\text{#2}}}
\newcommand{\mig}[1]{\phi_{\text{mig}}(#1)}
\newcommand{\nmig}[1]{\phi_{\neg\text{mig}}(#1)}

\newcommand{\enctr}[2]{[#1]\to[#2]}
\newcommand{\enctradd}[3]{[#1; {\fresh}]\to[#2; {#3}]}
\newcommand{\enctrstates}{\operatorname{states}}
\newcommand{\tran}{\operatorname{tran}}
\newcommand{\trgen}[1]{\operatorname{\textsc{#1}}}  

\newcommand{\autto}[1][\!\Aut]{\to_{#1}}
\newcommand{\readyto}[1][\text{ready}]{\Rightarrow_{#1}}
\newcommand{\chanzto}[1][\Chan^0]{\Rightarrow_{#1}}

\let\bisim\approx
\newcommand{\bisimA}{\sim}
\newcommand{\bisimP}{\backsim}

\newcommand{\Chan}{\mathcal{C}}

\makeatother

%% file: style/theorems.tex
\usepackage{amsthm}

\theoremstyle{plain}
\newtheorem{theorem}{Theorem}
\newtheorem{lemma}{Lemma}
\newtheorem{proposition}{Proposition}
\newtheorem{remark}{Remark}
\newtheorem{corollary}{Corollary}

\theoremstyle{definition}
\newtheorem{definition}{Definition}

\newtheorem*{nameuniq}{Name Uniqueness Assumption}
\makeatletter
\let\@nameuniq\nameuniq
\def\nameuniq{%
  \@nameuniq%
  \def\@currentlabel{Name Uniqueness}
  \phantomsection
  \label{nameuniq}
}
\makeatother

\theoremstyle{remark}
\newtheorem{example}{Example}

\newenvironment{proof*}
    {\begin{proof}[Proof \textup(Sketch\textup)]}
    {\end{proof}}

%% file: style/shortversion.tex
\newif\ifshortversion
\def\iflongversion{\ifshortversion\else} 

\def\shortversioninline#1{%
  \ifshortversion%
    $#1$
  \else%
    \[#1\]
  \fi%
}

\newif\ifincludeappendix
\includeappendixtrue

\def\appendixorfull{%
  \ifincludeappendix%
    Appendix%
  \else%
    \cite{fullversion}%
  \fi%
}

%% file: 1-intro.tex
\section{Introduction}
\label{sec:intro}

The \picalc\ \cite{MilnerPW92} is a concise yet expressive model of concurrent computation.
Its view of a concurrent system is a set of processes exchanging messages over channels, either private or public.
Both processes and private channels can be created dynamically.
A key feature of the calculus is mobility:
a private channel name can be sent as a message over a public one
and later used to exchange messages with an initially disconnected party.
The \emph{communication topology} of a \picalc\ system,
i.e.,~the graph linking processes that share channels,
is therefore dynamically evolving, in contrast to those of simpler process calculi such as CCS.

From a verification point of view,
proving properties of \picalc\ terms is challenging:
the full \picalc\ is Turing-complete.
As a consequence, a lot of research effort has been devoted to defining
fragments of \picalc\ that could be verified automatically
while retaining as much expressivity as possible.
To date, the most expressive fragment that has decidable verification problems is the \emph{depth-bounded \picalc}~\cite{meyer:db}.
Roughly speaking, the depth of a \picalc\ term can be understood
as the maximum length of the simple (i.e~non looping) paths
in the communication topology of the term.
A term is depth-bounded if there exists a $k \in \Nat$ such that the maximal nested depth of restriction of each reachable term is bounded by $k$.
Notably, depth-bounded systems can have an infinite state-space and generate unboundedly many names.
Besides enabling the design of procedures for deciding such important verification problems as termination or coverability,
depth boundedness can be useful as a correctness property of a system in itself.
Consider, for example, a system modelling an unbounded number of processes, each maintaining a private queue of tasks and communicating via message-passing.
In the \picalc, structures such as lists and queues are typically modelled using private channels to represent the ``next'' pointers.
Proving a bound in depth $k$ for such a system would guarantee that none of the queues grows unboundedly, which is an oft-desired resource-usage property.

Unfortunately, depth boundedness is a \emph{semantic} property,
it is undecidable whether a given arbitrary \picalc\ term is depth-bounded.
It has recently been proved that the problem becomes decidable if the bound $k$ is fixed~\cite{boundsmobility} but the complexity is very high.

\subsection*{Contributions}

The first contribution of this paper is a novel fragment of \picalc\ which we call \emph{typably hierarchical}, which is a proper subset of the depth-bounded \picalc.
This fragment is defined by means of a type system with decidable checking and inference.
The typably hierarchical fragment is rather expressive:
it includes terms that are unbounded in the number of private channels
and exhibit mobility.

The type system itself is based on the novel notion of \emph{\tcompat[ibility]}, where $\Types$ is a given finite forest.
We start from the observation that the communication topologies of depth bounded terms often exhibit a hierarchical structure:
channels are organisable into layers with decreasing degree of sharing.
Consider the example of an unbounded number of clients communicating with their local server:
a message from a client containing a private channel is sent to the server's channel, the server replies to the client's request on the client's private channel.
While the server's channel is shared among all the clients,
the private channel of each client is shared only between itself and the server.
\tcompat[ibility] formalises and generalises this intuition.
Roughly speaking, we associate to each channel name a \emph{base type}
which is a node in a (finite) forest $\Types$.
The forest $\Types$ represents the hierarchical relationship between channels:
it is the blueprint according to which one can organise the relationship between channels in each reachable term.
\ifshortversion%
  The type system we present exploits the structure of $\Types$ so that each term that is reachable from a typably hierarchical term is guaranteed to have
  a depth bounded by the height of $\Types$.
\else%
  \par
  More precisely, the names hierarchy imposes constraints on the scopes of private names that can be considered valid.
  Consider the term
    $(\new b.(\out a<b>.\inp b(y))) \parallel a(x).(\new c.\out x<c>)$:
  two parallel processes ready to synchronise on the public channel $a$.
  Upon synchronisation, the private name $b$%
  ---known only by the first process---%
  will be transmitted to the second process
  which will ``migrate'' under the scope of $b$.
  The result of this communication is the term
    $\new b.(\inp b(y) \parallel \new c.(\out b<c>))$,
  note how the migration nests the scope of $c$ in the scope of $b$.
  If $\Types$ dictates that $c$ is higher in the hierarchy than $b$
  the scoping resulting from the communication would be invalid:
  scope nesting should always respect the hierarchy.
  The type system we present constrains the use of names so that each term that is reachable from a typably hierarchical term is guaranteed to have
  scopes respecting $\Types$.
  From this guarantee it can be shown that typably hierarchical terms have
  a depth bounded by the height of $\Types$.
\fi%
We believe that the notion of \tcompat[ibility] has potential
as a specification device:
it allows the user to specify the desired relationship between channels instead of just a numeric bound on depth.

After defining the typably hierarchical fragment, we turn to the question:
\emph{is there an automata-based model that can represent the same set of systems?}
The second contribution of this paper is an encoding of typably hierarchical
into \emph{Nested Data Class Memory Automata}~\cite{conrad:14},
a class of automata over data-words (i.e.~finite words over infinite alphabets).
An encoding of Nested Data Class Memory Automata
into typably hierarchical terms is also presented,
showing that the two models are equi-expressive.
The two encodings are heavily based on the notion of \tcompat[ibility] and open an approach to fruitful interactions between process algebra and automata over infinite alphabets.

%% file: 2-picalc.tex
\section{Preliminaries}
\label{sec:picalc}\label{sec:prelim}

\subsection*{Labelled forests}

A \emph{forest} is a simple,
acyclic, directed graph $f = (N_f, \parent_f)$ such that the edge relation,
$\parent_f^{-1} \from N_f \pto N_f$, is the \emph{parent} map
which is defined on every node of the forest except the \emph{root}(s).
A \emph{path} is a sequence of nodes, $n_1 \, \dots \, n_k$,
such that for each $i < k$, $n_{i} \parent_f n_{i+1}$.
Thus every node of a forest has a unique path to a root
(and it follows that that root is unique).
Henceforth we assume that all forests are finite.
We write $\paths(f)$ for the set of paths in $f$.
The \emph{height} of a forest, $\height(f)$, is the length of its longest path.

An \emph{$L$-labelled forest} is a pair $\phi = (f_\phi, \ell_\phi)$ where
  $f_\phi$ is a forest and
  $\ell_\phi \from N_\phi \to L$ is a labelling function on nodes.
Given a path $n_1 \dots n_k$ of $f_\phi$, its \emph{trace} is the induced sequence
  $\ell_\phi(n_1) \dots \ell_\phi(n_k)$.
By abuse of language, a \emph{trace} is an element of $L^\ast$
which is the trace of some path in the forest.
We write $\traces(\phi)$ for the set of traces of the labelled forest.

We define $L$-labelled forests inductively from the empty forest $\emptyforest$.
We write $\phi_1\dunion \phi_2$ for the disjoint union
of forests $\phi_1$ and $\phi_2$,
and $l[\phi]$ for the forest with a single root, labelled with $l \in L$,
which has the respective roots of the forest $\phi$ as children.
\iflongversion
Since the choice of the set of nodes is irrelevant, we will always
interpret equality between forests up to isomorphism
(i.e. a bijection on nodes respecting parent and labeling).
\fi

\subsection*{The \picalc}
We use a \picalc{} with guarded replication to express recursion~\cite{milner:fun}.
Fix a universe $\Names$ of names representing channels and messages
occurring in communications.
The syntax follows the grammar:
\begin{grammar}
    \PiTerms \ni P, Q \is \new x.P | P_1 \parallel P_2 | M | \bang M
        && \text{process}\\
    M \is \zero | M + M | \pi_i. P_i
        && \text{choice}\\
    \pi \is \inp a(x) | \out a<b> | \tact
        && \text{prefix}
\end{grammar}
Structural congruence is defined as the smallest congruence closed by
  \pre\alpha-conversion of bound names
  commutativity and associativity of choice and parallel composition
  with $\zero$ as the neutral element,
  and the following laws for restriction, replication and scope extrusion:%
  \footnote{
    Technically, the $\bang \zero \congr \zero$ rule
    is not in the standard definition,
    but this does not affect the reduction semantics.%
  }%
\begin{mathpar}
  \new x.\zero \congr \zero \and
  \new x.\new y.P \congr \new y.\new x.P \and
  \bang \zero \congr \zero \and
  \bang M \congr M \parallel \bang M \and
   P \parallel \new a . Q \congr \new a . (P \parallel Q)
     \quad (\text{if } a \not\in \freenames(P))
\end{mathpar}

The name $x$ is bound in both $\new x.P$, and in $\inp a(x).P$.
We will write $\freenames(P)$, $\boundnames(P)$ and $\resboundnames(P)$
for the set of free, bound and restriction-bound names in $P$, respectively.
A sub-term is \emph{active} if it is not under a prefix.
A name is active when it is bound by an active restriction.
The set $\actrestr(P)$ is the set of the active names of $P$.
Terms of the form $M$ and $\bang M$ are called \emph{sequential}.
We write $\Seq$ for the set of all sequential terms.
$\seqproc(P)$ is the set of all active sequential processes of~$P$.

We will often rely on the following mild assumption,
that the choice of names is unambiguous,
especially when selecting a representative for a congruence class.

\begin{nameuniq}
  Each name in $P$ is bound at most once; and
  $\freenames(P) \inters \boundnames(P) = \emptyset$.
\end{nameuniq}

Note that channels are unary;
extending our work to the polyadic case is strightforward
but we only consider the unary case for conciseness.

As we will see in the rest of the paper,
the notions of depth and of hierarchy between names
rely heavily on structural congruence.
In particular, given a certain structure on names,
there will be a specific representative of the structural congruence class
that exhibits the desired properties.
Nevertheless, we cannot assume the input term
is always presented as that specific representative;
worse yet, when the structure on names is not fixed,
as in the case of type inference,
we cannot fix any particular representative and be sure
it will witness the desired properties.
So, instead, in the semantics and in the type system,
we manipulate a neutral representative called \emph{normal form},
which is a variant of the \emph{standard form}~\cite{milner:picalc}.
In this way we are not distracted by
the particular syntactic representation we are presented with.

We say that a term $P$ is in \emph{normal form} ($P \in \PiNf$) if it is in standard form
and each of its inactive subterms is also in normal form.
Formally, each process in normal form follows the grammar
\begin{grammar}
    \PiNf \ni N \is \new x_1.\cdots \new x_n.
                      (A_1 \parallel \cdots \parallel A_m) \\
    A \is \pi_1. N_1 + \cdots + \pi_n. N_n \\
      &|  \bang{\left(\pi_1. N_1 + \cdots + \pi_n. N_n \right)} \\
\end{grammar}
where the sequences $\lst{x}{n}$ and $\lst{A}{m}$ may be empty;
when they are both empty the normal form is the term $\zero$.
We further assume w.l.o.g. that a normal form satisfies \ref{nameuniq}.
Since the order of appearance of the restrictions, sequential terms or choices
in a normal form is irrelevant in the technical development of our results,
we use the following abbreviations.
Given a finite set of indexes $I = \set{i_1,\dots,i_n}$ we write
$\Parallel_{i \in I} A_i$ for $(A_{i_1} \parallel \cdots \parallel A_{i_n})$,
which is $\zero$ when $I$ is empty;
and $\Alt_{i \in I} \pi_i. N_i$ for
$(\pi_{i_1}. N_{i_1} + \cdots + \pi_{i_n}. N_{i_n})$.
This notation is justified by commutativity and associativity of the parallel and choice operators.
We also write $\new X.P$ or $\new x_1\:x_2\cdots x_n.P$
for $\new x_1.\cdots \new x_n.P$
when $X = \set{\lstc{x}{n}}$, or just $P$ when $X$ is empty;
this is justified by the structural laws of restrictions.
When $X$ and $Y$ are disjoint sets of names, we use juxtaposition for union.

\iflongversion
\begin{figure*}[tb]
  \centering
  \input{definitions/nf}
  \caption{Definition of the $\nf \from \PiTerms \to \PiNf$ function.}
  \label{fig:nf}
\end{figure*}
\fi

Every process $P \in \PiTerms$ is structurally congruent
to a process in normal form.
The function $\nf \from \PiTerms \to \PiNf$,
\ifshortversion
defined in~\appendixorfull,
\else
defined in Figure~\ref{fig:nf},
\fi
extracts, from a term, a structurally equivalent normal form.

We are interested in the reduction semantics of a \piterm{},
which can be described using the following rule.

\begin{definition}[Semantics of \picalc]
  \label{def:sem-of-picalc}
  The operational semantics of \picalc\ is defined
  by the transition system on \piterm{s},
  with transitions satisfying $P\redto Q$ if
  \begin{defenum}
    \item
      $P \congr \new W.(S \parallel R \parallel C) \in \PiNf$,
    \item
      $S = (\out a<b>.\new Y_s.S')+M_s$,
    \item
      $R = (\inp a(x).\new Y_r.R')+M_r$,
    \item
      $Q \congr \new W Y_s Y_r.
        (S' \parallel R'\subst{x->b} \parallel C)$,
  \end{defenum}
  or if
  \begin{defenum}
    \item
      $P \congr \new W.(\tact.\new Y.P' \parallel C) \in \PiNf$,
    \item
      $Q \congr \new W Y. (P' \parallel C)$.
  \end{defenum}

  We define the set of reachable configurations as
    $\reach(P) \is \Set{Q | P \redto^* Q}$,
  writing $\redto^*$ to mean the reflexive, transitive closure of $\redto$.
\end{definition}
Note that the use of structural congruence takes care of unfolding replications,
if necessary.

\begin{example}[Server/Client system]
  \label{ex:servers}
  Consider the term
    $\new s\:c.P$
  where:
  \begin{align*}
    P &= \bang{S} \parallel \bang{C} \parallel \bang{M} &
    S &= s(x).\new d.\out x<d> \\
    C &= c(m).(\out s<m> \parallel \inp m(y).\out c<m>) &
    M &= \tact.\new m.\out c<m>
  \end{align*}
  The term $\bang{S}$, which is presented in normal form, represents a server
  listening to a port $s$ for a client's requests.
  A request is a channel $x$ that the client sends to the server
  for exchanging the response.
  After receiving $x$ the server creates a new name $d$ and sends it over $x$.
  The term $\bang{M}$ creates unboundedly many clients,
  each with its own private mailbox $m$.
  A client on a mailbox $m$ repeatedly sends requests to the server
  and concurrently waits for the answer on the mailbox before recursing.
  An example run of the system:
  \begin{align*} 
    \new s\:c.P
    & \to \new s\:c\:m.(P \parallel \out c<m>) \\
    & \to \new s\:c\:m.(P \parallel \out s<m> \parallel \inp m(y).\out c<m>) \\
    & \to \new s\:c\:m\:d.(P \parallel \out m<d> \parallel \inp m(y).\out c<m>) \\
    & \to \new s\:c\:m\:d.(P \parallel \out c<m>)
      \congr \new s\:c\:m.(P \parallel \out c<m>)
  \end{align*}
\end{example}

\begin{example}[Stack-like system]
  \label{ex:stack}
  Consider the normal form
    $\new X.( \bang{S} \parallel \out s<a> )$
  where $X = \set{s, n, v, a}$ and
  \[
    S = \inp s(x).\new b.\bigl( (\out v<b>.\out n<x>) \parallel \out s<b> \bigr)
  \]
  The term $\out s<a>$ represents a stack with top element $a$;
  the stack is in an infinite loop that pushes new names (copies of $b$):
  this is represented by the term $\out v<b>.\out n<a> \parallel \out s<b>$
  indicating that the top value is $b$, the next is $a$
  and the stack now starts from $b$.
  An example run:
  \begin{align*}
    & \new X.( \bang{S} \parallel \out s<a> ) \\
    & \to \new X.( \bang{S} \parallel \new b.((\out v<b>.\out n<a>) \parallel \out s<b>)) \\
    & \to \new X.( \bang{S} \parallel \new b\:b'.((\out v<b>.\out n<a>) \parallel (\out v<b'>.\out n<b>) \parallel \out s<b'>))
  \end{align*}
\end{example}
The following definitions are minor variations of (but equivalent to) the concepts introduced in~\cite{meyer:db}.%
\footnote{%
  In~\cite{meyer:db} these functions are defined on fragments.
  It is easy to prove that our definition of $\nestr$ coincides
  with the one in~\cite{meyer:db} on fragments
  and that for any fragment $F$ and non-fragment $P$,
  if $F \congr P$ then $\nestr(P) \geq \nestr(F)$.
  As a consequence our definition of depth coincides
  with the one in~\cite{meyer:db}.
}
\begin{definition}[$\nestr$, $\depth$, depth-bounded term]
\label{def:nest}\label{def:depth}
\label{def:depth-bounded}
  The \emph{nesting of restrictions} of a term is given by the function
  \ifshortversion
  $\nestr \from \PiTerms \to \Nat$ defined as follows:
    $\nestr(M)  \is \nestr(\bang{M}) \is 0$,
    $\nestr(\new x.P)  \is 1 + \nestr(P)$,
    $\nestr(P \parallel Q)  \is \max(\nestr(P), \nestr(Q))$.
  \else
  \begin{align*}
    \nestr(M) & \is \nestr(\bang{M}) \is 0 \\
    \nestr(\new x.P) & \is 1 + \nestr(P) \\
    \nestr(P \parallel Q) & \is \max(\nestr(P), \nestr(Q)).
  \end{align*}
  \fi
  The \emph{depth} of a term is defined as the minimal nesting of restrictions
  in its congruence class: \shortversioninline{
    \depth(P) \is \min\set{\nestr(Q) | P \congr Q}.
  }
  A term $P \in \PiTerms$ is \emph{depth-bounded} if
  there exists a $k \in \Nat$ such that
  for each $Q \in \reach(P)$, $\depth(Q) \leq k$.
\end{definition}

\begin{example}
\label{ex:unbounded}
  The term in Example~\ref{ex:servers} is depth bounded:
  all the reachable terms are congruent to terms of the form
  \[
    Q_{ijk} = \new s\:c.\bigl(
      P \parallel N^i \parallel \mathit{Req}^j \parallel \mathit{Ans}^k
    \bigr)
  \]
  for some $i, j, k \in \Nat$
  where $N = \new m.\out c<m>$,
        $\mathit{Req} =\new m.(\out s<m> \parallel \inp m(y).\out c<m>)$,
        $\mathit{Ans} = \new m.(\new d.\out m<d> \parallel \inp m(y).\out c<m>)$
  and by $Q^n$ we mean the parallel composition of $n$ copies of the term $Q$.
  For any $i, j, k$, $\nestr(Q_{ijk}) \leq 4$:
  the longest chain of nested restrictions is $s, c, m, d$.

  The term in Example~\ref{ex:stack} is unbounded in depth:
  the number of nested copies of $b$ grows every time a push is performed;
  it is not possible to extrude their scope to reduce the number of nested levels.

  Note that both terms are not \emph{name bounded}
  (in the sense of~\cite{meyer:name}):
  the number of active restrictions in the reachable terms is not bounded.
\end{example}

\begin{definition}[Forest representation]
\label{def:forest-repr}
  We represent the structural congruence class of a term $P \in \PiTerms$
  with the set of labelled forests $\AST{P} \is \set{\forest(Q) | Q \congr P}$
  with labels in $\actrestr(P) \dunion \seqproc(P)$
  where $\forest(Q)$ is defined as
  \[
    \forest(Q) \is
      \begin{cases}
        x[\forest(Q')]                    \CASE Q = \new x.Q' \\
        \forest(Q_1) \dunion \forest(Q_2) \CASE Q = Q_1 \parallel Q_2 \\
        Q[\emptyforest]                   \CASE Q \text{ is sequential}\\
        \emptyforest                      \CASE Q = \zero
      \end{cases}
  \]
  Note that only leaves are labelled with sequential processes.

  The \emph{restriction height}, $\height_\restr(\forest(P))$, is the length
  of the longest path formed of nodes labelled with names in $\forest(P)$.
\end{definition}

Clearly, for any $P \in \PiTerms$,
$\depth(P) = \min\set{\height_\restr(\phi) | \phi \in \AST{P}}$.

\begin{lemma}
\label{lemma:forest-nf}
Let $\phi$ be a forest with labels in $\Names \dunion \Seq$.
Then $\phi = \forest(Q)$ with
\ifshortversion
    $Q \congr Q_\phi \is \new X_\phi.\Parallel_{(n, A) \in I} A$
  where
    $X_\phi \is \set{\ell_\phi(n) \in \Names | n \in N_\phi}$ and
    $I = \set{(n, A) \mid \ell_\phi(n) = A \in \Seq}$
\else
  $Q \congr Q_\phi$ where
  \begin{align*}
    Q_\phi &\is \new X_\phi.\Parallel_{(n, A) \in I} A\\
    X_\phi &\is \set{\ell_\phi(n) \in \Names | n \in N_\phi}\\
    I &\is \set{(n, A) \mid \ell_\phi(n) = A \in \Seq}
  \end{align*}
\fi
provided
\begin{enumerate}
  \item $\forall n \in N_\phi$,
        if $\ell_\phi(n) \in \Seq$ then
        $n$ has no children in $\phi$, and
        \label{lemma:forest-nf:seq-leaf}
  \item $\forall n, n' \in N_\phi$,
        if $\ell_\phi(n) = \ell_\phi(n') \in \Names$
        then $n = n'$, and
        \label{lemma:forest-nf:name-uniq}
  \item $\forall n \in N_\phi$,
        if $\ell_\phi(n) = A \in \Seq$
        then for each $x \in X_\phi \inters \freenames(A)$
          there exists $n' <_\phi n$ such that $\ell_\phi(n') = x$.
          \label{lemma:forest-nf:scoping}
\end{enumerate}
\end{lemma}

\begin{proof}
  \ifshortversion
    \input{proofs/forest-nf-sketch}
  \else
    \input{proofs/forest-nf}

  \fi
\end{proof}

%% file: definitions/nf.tex
\begin{mathpar}
  \nf(\zero)  \is \zero \and
  \nf(\pi.P)  \is \pi.\nf(P) \and
  \nf(\new x.P) \is \new x. \nf(P) \and
  \nf(M + M') \is
      \begin{cases}
        \nf(M)  \CASE \nf(M') = \zero \neq \nf(M) \\
        \nf(M') \CASE \nf(M) = \zero \\
        \nf(M) + \nf(M')  \OTHERWISE
      \end{cases}
      \and
  \nf(\bang{M}) \is
      \begin{cases}
        \bang{(\nf(M))} \CASE \nf(M) \neq \zero \\
        \zero \OTHERWISE
      \end{cases}
    \and
  \nf(P \parallel Q) \is
      \begin{cases}
        \nf(P)  \CASE \nf(Q) = \zero \neq \nf(P) \\
        \nf(Q)  \CASE \nf(P) = \zero \\
        \new X_P X_Q. (N_P \parallel N_Q) \CASE
          \nf(Q) = \new X_Q.N_Q,
          \nf(P) = \new X_P.N_P
          \AND
          \actrestr(N_P) = \actrestr(N_Q) = \emptyset
      \end{cases}
\end{mathpar}

%% file: proofs/forest-nf-sketch.tex
By induction on the structure of $\phi$.

%% file: proofs/forest-nf.tex
\newcommand{\ALLCOND}{%
  \ref{lemma:forest-nf:seq-leaf},
  \ref{lemma:forest-nf:name-uniq} and
  \ref{lemma:forest-nf:scoping}
}
We proceed by induction on the structure of $\phi$.
The base case is when $\phi = \emptyforest$,
for which we have $Q_\phi = \zero$ and $\phi = \forest(\zero)$.

When $\phi = \phi_0 \dunion \phi_1$ we have that
if conditions \ALLCOND hold for $\phi$,
they must hold for $\phi_0$ and $\phi_1$ as well,
hence we can apply the induction hypothesis to them obtaining
$\phi_i\forest(Q_i)$ with $Q_i \congr Q_{\phi_i}$ ($i \in \set{0,1}$).
We have $\phi = \forest(Q_0 \parallel Q_1)$ by definition of $\forest$,
and we want to prove that $Q_0 \parallel Q_1 \congr Q_{\phi}$.
By condition \ref{lemma:forest-nf:name-uniq} on $\phi$,
$X_{\phi_0}$ and $X_{\phi_1}$ must be disjoint;
furthermore, by condition \ref{lemma:forest-nf:scoping}
on both $\phi_0$ and $\phi_1$ we can conclude that
$\freenames(Q_{\phi_i}) \inters X_{\phi_{1-i}} = \emptyset$.
We can therefore apply scope extrusion:
$Q_0 \parallel Q_1
 \congr Q_{\phi_0} \parallel Q_{\phi_1}
 \congr \new X_{\phi_0} X_{\phi_1}.(P_{\phi_0} \parallel P_{\phi_1})
 = Q_\phi$.

The last case is when $\phi = l[\phi']$.
Suppose conditions \ALLCOND hold for $\phi$.
We distinguish two cases.
If $l = A \in \Seq$, by \ref{lemma:forest-nf:seq-leaf}
we have $\phi' = \emptyforest$, $\phi = \forest(A)$ and $A = Q_{\phi}$.
If $l = x \in \Names$ then we observe that conditions \ALLCOND hold for $\phi'$
under the assumption that they hold for $\phi$.
Therefore $\phi' = \forest(Q')$ with $Q' \congr Q_{\phi'}$,
and, by definition of $\forest$, $\phi = \forest(\new x.Q')$.
By condition \ref{lemma:forest-nf:name-uniq} we have $x \not\in X_{\phi'}$ so
$\new x.Q'
 \congr \new x. Q_{\phi'}
 \congr \new (X \union \set{x}).P_{\phi'}
 = Q_\phi$.

%% file: 3-tcompat.tex
\section{The notion of \tcompat[ibility]}
\label{sec:t-compat}

In this section we will introduce the concept of \emph{\tcompat[ibility]},
which is a central tool in our constructions.
First we will introduce types, which annotate names,
and postulate that they are arranged as a forest $(\Types, \parent)$.
Intuitively, by annotating names with types we impose a hierarchy on them,
and \tcompat[ibility] of a term $P$ will mean
that the structure of $P$ respects this hierarchy.

For the rest of the paper we will fix a
  \emph{finite forest of base types} $(\Types, \parent)$ where $n_1 \parent n_2$ means that ``$n_1$ is the parent of $n_2$''.
We write $\tleq$ and $\tlt$ for the
reflexive transitive and the transitive closure
of $\parent$, respectively.

Types are of the form
\begin{grammar}
    \type \is t \mid t[\type]
\end{grammar}
where $t \in \Types$ is a base type.
A name with type $t$ cannot be used as a channel but can be used as a message;
a name with type $t[\type]$ can be used to transmit a name of type $\type$.
We will write $\base(\type)$ for $t$ when $\type = t[\type']$ or $\type = t$.
Note that these are (a fragment of) the I/O-types in the sense of Pierce and Sangiorgi \cite{PierceS93}.
An environment $\Env$ is a partial map from names to types,
which we will write as a set of \emph{type assignments}, $x \tas \type$.
Given a set of names $X$ and an environment $\Env$,
we write $\Env(X)$ for the set $\set{\Env(x) | x \in X \inters \domain(\Env)}$.
Given two environments $\Env$ and $\Env'$ with
  $\domain(\Env)\inters\domain(\Env') = \emptyset$,
we write $\Env \Env'$ for their union.
For a type environment $\Env$ we define
$\minrestr(\Env) \is \set{(x \tas \type) \in \Env |
    \forall (y \tas \type') \in \Env \st \base(\type') \not\tlt \base(\type)}$.

\iflongversion
\begin{figure*}
  \centering
  \tikzsetnextfilename{example-forests}
  \begin{tikzpicture}
    \makeatletter
    \def\forestname#1#2{\def\@currentlabel{#1}\phantomsection\label{forest:#2}}
    \makeatother
    \forestname{1}{broom}
    \forestname{2}{cab}
    \forestname{3}{min-abc}
    \forestname{4}{min-ab-c}
    \forestname{5}{min-ba-c}
    \input{figures/example-forests.tikz}
  \end{tikzpicture}
  \caption{%
    Examples of forests in $\AST{P}$ of Example~\ref{ex:tied-to}:
    $ P = \new a\:b\:c. ( A_1 \parallel A_2 \parallel A_3 \parallel A_4 ) $
    where
      $A_1 = \inp a(x)$,
      $A_2 = \inp b(x)$,
      $A_3 = \inp c(x)$ and
      $A_4 = \out a<b>$.%
  }
  \label{fig:forests}
\end{figure*}
\fi

From now on, we will assume every \piterm\ is annotated with types:
in a restriction $\restr X$, $X$ is a set of type assignments.

\begin{definition}[Annotated term]
\label{def:annot-term}
  A \emph{$\Types$-annotated \piterm} (or simply \emph{annotated \piterm}) $P \in \PiAnnot$
  has the same syntax as regular \piterm{s}
  except restrictions take the form $\restr x\tas \type$.
  The semantics is the same, except type annotations get copied
  when a name is duplicated or renamed by structural congruence.
  The definition of forest representation is also extended
  to annotated \piterm{s} by changing the case
  when $Q = \new x\tas \type.Q'$
  to $(x, t)[\forest(Q')]$, where $\base(\type) = t$.
  The forests in $\AST{P}$ will thus have labels in
  $(\actrestr(P) \times \Types) \dunion \seqproc(P)$.
  We write $\Forests_{\Types}$ for the set of forests with labels in
    $(\Names \times \Types) \dunion \Seq$.
  The set $\PiNfAnnot$ contains all the annotated \piterm{s} in normal form.
\end{definition}

Given a normal form $P = \new X. \Parallel_{i \in I} A_i$ we say that
$A_i$ is \emph{linked to $A_j$ in $P$}, written $i \linkedto{P} j $, if
$\freenames(A_i) \inters \freenames(A_j) \inters \set{x | (x \tas \type) \in X}
    \neq \emptyset$.
We also define the \emph{tied-to} relation as
the transitive closure of $\linkedto{P}$.
I.e.~$A_i$ is \emph{tied to} $A_j$, written $i \tiedto{P} j$, if
$\exists k \in I \st i \linkedto{P} k \, \wedge \, k \tiedto{P} j$.
Furthermore, we say that a name $y$ is \emph{tied to $A_i$ in $P$},
written $y \ntiedto{P} i$,
if $\exists j \in I \st y \in \freenames(A_j) \, \wedge \, j \tiedto{P} i$.
Given an input-prefixed normal form $\inp a(y). P$
where $P = \new X. \Parallel_{i \in I} A_i$,
we say that \emph{$A_i$ is migratable in $\inp a(y). P$},
written $\migr{\inp a(y).P}(i)$, if
$y \ntiedto{P} i$.

The tied-to relation may seem obscure at first.
Its meaning is better explained by the following lemma
which indicates how this relation fundamentally constrains the possible shape of the forest of a term.

\begin{lemma}
\label{lemma:tied-tree}
  Let $P = \new X.\Parallel_{i\in I} A_i \in \PiNfAnnot$, if $i \tiedto{P} j$
  then any forest $\phi \in \AST{P}$ containing two leaves
  labelled with $A_i$ and $A_j$ respectively,
  will be such that these leaves belong to the same tree
  (i.e.~have a common ancestor in $\phi$).
\end{lemma}

\iflongversion
\begin{proof}
  \input{proofs/tied-tree}
\end{proof}

\begin{example}
\label{ex:tied-to}
  Take the normal form
    $ P = \new a\:b\:c. ( A_1 \parallel A_2 \parallel A_3 \parallel A_4 ) $
    where
      $A_1 = \inp a(x)$,
      $A_2 = \inp b(x)$,
      $A_3 = \inp c(x)$ and
      $A_4 = \out a<b>$.
  We have $1 \linkedto{P} 4$, $2 \linkedto{P} 4$,
  therefore $1 \tiedto{P} 2 \tiedto{P} 4$ and $a \ntiedto{P} 2$.
  In Figure~\ref{fig:forests} we show some of the forests in $\AST{P}$.
  Forest~\ref{forest:broom} represents $\forest(P)$.
  The fact that $A_1, A_2$ and $A_4$ are tied is reflected by the fact
  that none of the forests place them in disjoint trees.
  Now suppose we select only the forests in $\AST{P}$
  that have $a$ as an ancestor of $b$: in all the forests in this set,
  the nodes labelled with $A_1, A_2$ and $A_4$ have $a$ as common ancestor
  (as in forests~\ref{forest:broom},
                 \ref{forest:cab},
                 \ref{forest:min-abc} and
                 \ref{forest:min-ab-c}).
  In particular, in these forests $A_2$ is necessarily a descendent of $a$
  even if $a$ is not one of its free names.
\end{example}
\fi

\begin{definition}[{\tcompat[ibility]}]
\label{def:t-compat}
  Let $P \in \PiAnnot$ be an annotated \piterm.
  A forest $\phi \in \AST{P}$ is said to be \emph{\tcompat} if
  for every trace
    $((x_1, t_1) \dots (x_k, t_k) \: A)$
  in $\phi$ it holds that
    $t_1 \tlt t_2 \dots \tlt t_k$.
  $P$ is said to be \emph{\tcompat} if there exists
  a \tcompat\ forest in $\AST{P}$.
  A term is \emph{\tshaped}\ if each of its subterms is \tcompat.
\end{definition}

\begin{figure*}[!t]
  \input{definitions/phi}
  \caption[Definition of $\Phi_\Types$]{%
    Definition of $\Phi_\Types \from \PiNfAnnot \to \Forests_\Types$.%
  }
  \label{fig:phi}
\end{figure*}

\begin{example}
  \label{ex:types}
  Let us fix $\Types$ to be the forest
    $\ty{s} \parent \ty{c} \parent \ty{m} \parent \ty{d}$.
  The normal form in Example~\ref{ex:servers} is \tcompat{}
  when $s$ and $c$ are annotated with types $\type_s$ and $\type_c$
  respectively, with $\base(\type_s) = \ty{s}$ and $\base(\type_c)=\ty{c}$;
  indeed we have
    $\forest(\new (s\tas\type_s)\:(c\tas\type_c).P)
      = (s, \ty{s})\bigl[
          (c, \ty{c})[\:
            \bang{S}[] \dunion \bang{C}[] \dunion \bang{M}[]
          \:]
        \bigr]$.
  By annotating $m$ and $d$ with types with base type $\ty{m}$ and $\ty{d}$
  respectively, the term is also \tshaped.
\end{example}

Since \tcompat[ibility] is a condition on types,
\pre\alpha-renaming does not interfere with it.

\begin{lemma}
\label{lemma:tcompat-alpha}
  If $\forest(P)$ is \tcompat\ then
  for any term $Q$ which is an \pre\alpha-renaming of $P$,
  $\forest(Q)$ is \tcompat.
\end{lemma}

\begin{lemma}
\label{lemma:tcompat-takeout}
  Let $P = \new X.\Parallel_{i \in I} A_i$ be a \tcompat\ normal form,
  $Y \subseteq X$ and $J \subseteq I$.
  Then $P' = \new Y.\Parallel_{j \in J} A_j$ is \tcompat.
\end{lemma}

\iflongversion
\begin{proof}
  \input{proofs/tcompat-takeout}
\end{proof}
\fi

While many forests in $\AST{P}$ can be witnesses of the \tcompat[ibility] of $P$,
we want to characterise the shape of a witness that \emph{must} exist
if $P$ is \tcompat.
Such forest is identified by $\Phi_\Types(\nf(P))$ where
$\Phi_\Types \from \PiNfAnnot \to \Forests_\Types$
is the function defined in Figure~\ref{fig:phi}.
We omit the subscript when irrelevant or clear from the context.

\begin{example}
  In the run shown in Example~\ref{ex:servers}, after three steps we reach
  $Q = \new s\:c\:m\:d.(P \parallel \out m<d> \parallel \inp m(y).\out c<m>)$.
  The forest $\Phi_\Types(Q)$, when $\Types$ and types annotations are as in Example~\ref{ex:types}, is
  \begin{center}\small
  \tikzsetnextfilename{example-phi}
  \begin{tikzpicture}[scale=.8]
    \input{figures/example-phi.tikz}
  \end{tikzpicture}
  \end{center}
  where the nodes show only the name components of their labels for conciseness.
  Note how the scope of names is minimised while respecting \tcompat[ibility].

  Consider the term $P$ in Example~\ref{ex:tied-to}, with annotations
  $a \tas \ty{a}[\ty{b}[t]]$,
  $b \tas \ty{b}[t]$ and
  $c \tas \ty{c}[t']$.
  Forests \ref{forest:min-ab-c} and \ref{forest:min-ba-c}
  of Figure~\ref{fig:forests} represent $\Phi_\Types(P)$
  when $\Types$ is
    $\ty{a} \parent \ty{b}$ and
    $\ty{b} \parent \ty{a}$ respectively.
\end{example}

\begin{lemma}
\label{lemma:phi-tcompat}
  Let $P \in \PiNfAnnot$.
  Then:
  \begin{enumerate}[label=\alph*)]
    \item $\Phi_\Types(P)$ is a \tcompat\ forest;
          \label{lemma:phi-tcompat:tcompat}
    \item $\Phi_\Types(P) \in \AST{P}$ if and only if $P$ is \tcompat;
          \label{lemma:phi-tcompat:ast}
    \item if $P \congr Q \in \PiAnnot$ then
          $\Phi_\Types(P) \in \AST{Q}$ if and only if $Q$ is \tcompat.
          \label{lemma:phi-tcompat:congr}
  \end{enumerate}
\end{lemma}

\iflongversion
\begin{proof}
  \input{proofs/phi-tcompat}
\end{proof}
\fi

\begin{lemma}
\label{lemma:phi-tied}
  Let $P = \new X . \Parallel_{i \in I} A_i \in \PiNfAnnot$ be a \tcompat\ normal form.
  Then for every trace
    $((x_1, t_1) \dots (x_k, t_k) \: A_j)$
  in the forest $\Phi(P)$,
  for every $i \in \set{1, \ldots, k}$, we have $x_i \ntiedto{P} j$ (i.e.
  $x_i$ is tied to $A_j$ in $P$).
\end{lemma}

\begin{proof}
  Straightforward from the definition of $I_x$ in $\Phi$:
  when a node labelled by $(x,t)$ is introduced,
  its subtree is extracted from a recursive call on a term
  that contains all and only the sequential terms that are tied to $x$.
\end{proof}

\iflongversion
\begin{remark}
\label{remark:phi-nf}
  $\Phi(P)$ satisfies conditions
    \ref{lemma:forest-nf:seq-leaf},
    \ref{lemma:forest-nf:name-uniq} and
    \ref{lemma:forest-nf:scoping}
  of Lemma~\ref{lemma:forest-nf}.
\end{remark}
\fi

It is clear from the definition that if a \piterm{} $P$ is \tcompat{}
then $\depth(P)$ is bounded by the length of the longest strictly increasing chain in \Types; since \Types\ is assumed to be finite, the bound on the depth is finite.

\begin{proposition}
\label{prop:t-compat-db}
  Let $\Types$ be a forest and $P$ an annotated \piterm.
  If every $Q \in \reach(P)$ is \tcompat,
  then $P$ is depth-bounded.
\end{proposition}

\begin{example}
  Fix $\Types$ to be the forest
    $\ty{n} \parent \ty{v} \parent \ty{s} \parent \ty{a}$
  and take the term of Example~\ref{ex:stack} annotating it with types
  such that the base types of the names $n, v, s, a$ and $b$ are
  $\ty{n}, \ty{v}, \ty{s}, \ty{a}$ and $\ty{a}$ respectively.
  The term $\new n\:v\:s\:a.( \bang{S} \parallel \out s<a> )$ is \tcompat,
  but the term
    $Q = \new n\:v\:s\:a\:b\:b'.(
      \bang{S} \parallel
      (\out v<b>.\out n<a>) \parallel
      (\out v<b'>.\out n<b>) \parallel
      \out s<b'>
    )$,
  reachable from it, is not:
  $b$ and $b'$ have the same base type $\ty{a}$ but need to be in the same trace
  in any forest of $\AST{Q}$.
  As we have shown in Example~\ref{ex:unbounded},
  this term is not bounded in depth, so there cannot be any finite $\Types$
  such that every reachable term is \tcompat.
\end{example}

%% file: figures/example-forests.tikz
\makeatletter
\def\forestname#1#2{
  #1\def\@currentlabel{#1}\phantomsection\label{forest:#2}
}
\makeatother
\ensuretikzpicturebegin
\begin{scope}[
  AST,
  forests distance=3.4
]
\def\newforest{++(\forestsdist,0) node}
\def\bratio{.43}

\path

\newforest (broom) {$a$}
  child[sibling distance=6mm] { node {$b$}
    child { node {$c$}
      child { node {$A_1$}}
      child { node {$A_2$}}
      child { node {$A_3$}}
      child { node {$A_4$}}
    }
  }

\newforest (cab) {$c$}
  child { node {$a$}
    child { node {$A_1$}}
    child { node {$b$}
      child { node {$A_2$}}
      child { node {$A_4$}}
    }
    child { node {$A_3$}}
  }

\newforest (min-abc) {$a$}
  child { node {$A_1$}}
  child { node {$b$}
    child { node {$A_2$}}
    child { node {$c$}
      child { node {$A_3$}}
    }
    child { node {$A_4$}}
  }
  child[missing] {}

++(-.2,0)
\newforest {$a$}
  child { node {$A_1$}}
  child { node {$b$}
    child { node {$A_2$}}
    child { node {$A_4$}}
  }
  child[missing] {}

++(.2,0) coordinate (min-ab-c) +(.8,0)
node {$c$}
    child { node {$A_3$}}

++(-.2,0)
\newforest {$b$}
  child { node {$A_2$}}
  child { node {$a$}
    child { node {$A_1$}}
    child { node {$A_4$}}
  }
  child[missing] {}

++(.2,0) coordinate (min-ba-c) +(.8,0)
node {$c$}
    child { node {$A_3$}}

;

\path[draw=black!20, rounded corners]
  foreach[count=\i] \n in {broom,cab,min-abc,min-ab-c,min-ba-c}{
    (\n) ++(-\bratio*\forestsdist,.3)
      node[circle,inner sep=1pt, fill=white, draw] at +(.8ex,-.8ex) {\forestname{\i}{\n}}
      rectangle +(2*\bratio*\forestsdist,-2.4)
  }
;
\end{scope}
\ensuretikzpictureend

%% file: proofs/tied-tree.tex
We show that the claim holds in the case where $A_i$ is linked to $A_j$ in $P$.
From this, a simple induction over the length of linked-to steps required
to prove $i \tiedto{P} j$, can prove the lemma.

Suppose $i \linkedto{P} j$.
Let $Y = \freenames(A_i)
           \inters
         \freenames(A_j)
           \inters
         \set{x | (x \tas \type) \in X}$,
we have $Y \neq \emptyset$.
Both $A_i$ and $A_j$ are in the scope
of each of the restrictions bounding names $y \in Y$ in any
of the processes $Q$ in the congruence class of $P$,
hence, by definition of $\forest$, the nodes labelled with $A_i$ and $A_j$
generated by $\forest(Q)$ will have nodes labelled with $(y, \base(X(y)))$
as common ancestors.

%% file: definitions/phi.tex
\begin{equation*}
  \Phi_\Types(\new X.\Parallel_{i \in I} A_i) \is
    \begin{cases}
      \Dunion_{i\in I} \set{A_i[]} \CASE X = \emptyset \\
        \left(
          \Dunion
            \set{(x,\base(\type))[ \Phi_\Types(\new Y_x.\Parallel_{j \in I_x} A_j) ]
                | (x\tas\type) \in \minrestr(X) }
        \right)
      \dunion
      \Phi_\Types(\new Z.\Parallel_{r\in R} A_r)
      \CASE X \neq \emptyset
    \end{cases}
\end{equation*}
where
\ifshortversion%
  $P = \new X.\Parallel_{i \in I} A_i$,
  $I_x = \set{i\in I | x \ntiedto{P} i}$,
  $R = I \setminus (\Union_{(x\tas\type) \in \minrestr(X)} I_x)$,
  $Y_x = \set{(y \tas \type) \in X | \exists i \in I_x \, .\, y \in \freenames(A_i)} \setminus \minrestr(X)$ and
  $Z = X \setminus (\Union_{(x\tas\type) \in \minrestr(X)} Y_x \union \set{x \tas \type})$.
\else%
    $P = \new X.\Parallel_{i \in I} A_i$ and
  \begin{align*}
    I_x &= \set{i\in I | x \ntiedto{P} i} &
    R &= I \setminus (\textstyle \Union_{(x\tas\type) \in \minrestr(X)} I_x)\\
    Y_x &= \set{(y \tas \type) \in X | \exists i \in I_x \, .\, y \in \freenames(A_i)} \setminus \minrestr(X) &
    Z &= X \setminus (\textstyle \Union_{(x\tas\type) \in \minrestr(X)} Y_x \union \set{x \tas \type})
  \end{align*}
\fi%

%% file: proofs/tcompat-takeout.tex
  Take a \tcompat\ forest $\phi \in \AST{P}$.
  By Lemma~\ref{lemma:tcompat-alpha} we can assume without loss of generality
  that $\phi = \forest(Q)$ where proving $Q \congr P$
  does not require \pre\alpha-renaming.
  Clearly, removing the leaves that do not correspond to sequential terms
  indexed by $Y$ does not affect the \tcompat[ibility] of $\phi$.
  Similarly, if a restriction $(x\tas\type)\in X$ is not in $Y$,
  we can remove the node of $\phi$ labelled with $(x, \base(\type))$
  by making its parent the new parent of its children.
  This operation is unambiguous under \mbox{\ref{nameuniq}}
  and does not affect \tcompat[ibility], by transitivity of $\tlt$.
  We then obtain a forest $\phi'$ which is \tcompat\ and that,
  by Lemma~\ref{lemma:forest-nf},
  is the forest of a term congruent to the desired normal form $P'$.

%% file: figures/example-phi.tikz
\ensuretikzpicturebegin
\begin{scope}[AST]

\node (s) {$s$}
  child { node {$c$}
    child { node {$m$}
      [sibling distance=15mm]
      child { node {$\inp m(y).\out c<m>$} }
      child { node {$d$}
        child { node {$\out m<d>$} }
      }
    }
    child[missing] { }
    child { node {$\bang{S}$}}
    child { node {$\bang{C}$}}
    child { node {$\bang{M}$}}
  };

\end{scope}
\ensuretikzpictureend

%% file: proofs/phi-tcompat.tex
Item~\ref{lemma:phi-tcompat:tcompat} is an easy induction on the cardinality of $X$.

Item~\ref{lemma:phi-tcompat:ast} requires more work.
By item~\ref{lemma:phi-tcompat:tcompat} $\Phi(P)$ is \tcompat\ so
$\Phi(P) \in \AST{P}$ proves that $P$ is \tcompat.

To prove the $\Leftarrow$-direction we assume that
$P = \new X.\Parallel_{i\in I} A_i$ is \tcompat{}
and proceed by induction on the cardinality of $X$
to show that $\Phi(P) \in \AST{P}$.
The base case is when $X = \emptyset$:
$ \Phi(P) = \Phi(\Parallel_{i\in I} A_i)
          = \Dunion_{i\in I} \set{A_i[]}
          = \forest(\Parallel_{i\in I} A_i)
          = \forest(P) \in \AST{P}
$.
For the induction step,
we observe that $X \neq \emptyset$ implies $\minrestr(X) \neq \emptyset$ so,
$Z \subset X$ and
for each $(x \tas \type) \in \minrestr(X)$, $Y_x \subset X$
since $x \not\in Y_x$.
This, together with Lemma~\ref{lemma:tcompat-takeout},
allows us to apply the induction hypotesis on the terms
  $P_x = \new Y_x.\Parallel_{j \in I_x} A_j$ and
  $P_R = \new Z.\Parallel_{r\in R} A_r$,
obtaining that there exist terms $Q_x \congr P_x$ and $Q_R \congr P_R$ such that
$\forest(Q_x) = \Phi(P_x)$ and $\forest(Q_R) = \Phi(P_R)$ where
all the forests $\forest(Q_x)$ and $\forest(Q_R)$ are \tcompat.
Let
$ Q = \Parallel \set{\new (x \tas \type). Q_x | (x \tas \type) \in \minrestr(X)}
      \parallel Q_R
$,
then $\forest(Q) = \Phi(P)$.
To prove the claim we only need to show that $Q \congr P$.
We have
$ Q \congr \Parallel
      \set{ \new (x \tas \type).\new Y_x.\Parallel_{j \in I_x} A_j
            | (x \tas \type) \in \minrestr(X) }
      \parallel P_R
$
and we want to apply extrusion to get
$ Q \congr
      \new Y_{\min}.
      \left( \Parallel_{i \in I_{\min}} A_i \right)
      \parallel P_R
$
for
  $I_{\min} = \Dunion \set{I_x | (x \tas \type) \in \minrestr(X)}$,
  $Y_{\min} = \minrestr(X) \dunion \Dunion \set{Y_x | (x \tas \type) \in \minrestr(X)}$
which adds an obligation to prove that
\begin{enumerate}[label=\roman*)]
  \item $I_x$ are all pairwise disjoint so that $I_{\min}$ is well-defined,
        \label{phi-tcompat:Ix-disjoint}
  \item $Y_x$ are all pairwise disjoint and all disjoint from $\minrestr(X)$
        so that $Y_{\min}$ is well-defined,
        \label{phi-tcompat:Yx-disjoint}
  \item $Y_x \inters \freenames(A_j) = \emptyset$ for every $j \in I_z$
        with $z \neq x$ so that we can apply the extrusion rule.
        \label{phi-tcompat:Yx-Iz}
\end{enumerate}

To prove condition~\ref{phi-tcompat:Ix-disjoint},
assume by contradiction that there exists an $i \in I$
and names $x, y \in \minrestr(X)$ with $x \neq y$,
such that both $x$ and $y$ are tied to $A_i$ in $P$.
By transitivity of the tied-to relation, we have $I_x = I_y$.
By Lemma~\ref{lemma:tied-tree} all the $A_j$ with $j \in I_x$
need to be in the same tree in any forest $\phi \in \AST{P}$.
Since $P$ is \tcompat\ there exist such a $\phi$ which is \tcompat\ and
has every $A_j$ as label of leaves of the same tree.
This tree will include a node $n_x$ labelled with $(x, \base(X(x)))$
and a node $n_y$ labelled with $(y, \base(X(y)))$.
By \tcompat[ibility] of $\phi$
and the existence of a path between $n_x$ and $n_y$
we infer $\base(X(x)) < \base(X(y))$ or $\base(X(y)) < \base(X(x))$
which contradicts the assumption that $x, y \in \minrestr(X)$.

Condition~\ref{phi-tcompat:Yx-disjoint} follows from condition~\ref{phi-tcompat:Ix-disjoint}:
suppose there exists a $(z \tas \type) \in X \inters Y_x \inters Y_y$ for $x \neq y$,
then we would have that $z \in \freenames(A_i) \inters \freenames(A_j)$
for some $i \in I_x$ and $j \in I_y$,
but then $i \tiedto{P} j$,
meaning that $i \in I_y$ and $j \in I_x$ violating condition~\ref{phi-tcompat:Ix-disjoint}.
The fact that $Y_x \inters \minrestr(X) = \emptyset$
follows from the definition of $Y_x$.
The same reasoning proves condition~\ref{phi-tcompat:Yx-Iz}.

Now we have
$ Q \congr
      \new Y_{\min}.
      \left( \Parallel_{i \in I_{\min}} A_i \right)
      \parallel \new Z.\Parallel_{r\in R} A_r $
and we want to apply extrusion again to get
$ Q \congr
      \new Y_{\min} Z. \Parallel \set{A_i | i \in (I_{\min} \dunion R)}
$
which is sound under the following conditions:
\begin{enumerate}[resume*]
  \item $Y_{\min} \inters Z = \emptyset$,
  \item $I_{\min} \inters R = \emptyset$,
  \item $Z \inters \freenames(A_i) = \emptyset$ for all $i \not\in R$
\end{enumerate}
of which the first two hold trivially by construction,
while the last follows from condition~\ref{phi-tcompat:X-complete} below,
as a name in the intersection of $Z$ and a $\freenames(A_i)$
would need to be in $X$ but not in $Y_{\min}$.
To be able to conclude that $Q \congr P$ it remains to prove that
\begin{enumerate}[resume*]
  \item $I = I_{\min} \dunion R$ and
        \label{phi-tcompat:I-complete}
  \item $X = Y_{\min} \dunion Z$
        \label{phi-tcompat:X-complete}
\end{enumerate}
which are also trivially valid by inspection of their definitions.
This concludes the proof for item~\ref{lemma:phi-tcompat:ast}.

Finally, for every $Q \in \PiAnnot$ such that
$Q \congr P$, $\Phi(P) \in \AST{Q}$ if and only if $\Phi(P) \in \AST{P}$
by definition of $\AST{-}$;
since $\Phi(P)$ is \tcompat\ we can infer that
$Q$ is \tcompat{} if and only if $\Phi(P) \in \AST{Q}$,
which proves item~\ref{lemma:phi-tcompat:congr}.

%% file: 4-typesys.tex
\section{A type system for hierarchical topologies}
\label{sec:typesys}

We now define a type system to prove depth boundedness.
Our goal is to use Proposition~\ref{prop:t-compat-db}
by devising a type system, parametrised over $\Types$,
such that typability implies invariance of \tcompat[ibility] under reduction.
Typability of a \tcompat\ term $P$ would then imply that
every term reachable from it is \tcompat, entailing depth boundedness of $P$.

A judgement $\Env\types P$ means that $P \in \PiNfAnnot$ can be typed under assumptions $\Env$, over the tree \Types;
we say that $P$ is \emph{typable} if $\Env\types P$ is provable for some $\Env$ and $\Types$.
An arbitrary term $P \in \PiAnnot$ is said to be \emph{typable} if its normal form is.
The typing rules are presented in Figure~\ref{fig:typesys}.

\begin{figure*}[!t]
  \centering
  \input{definitions/typesys}
  \caption{
    A type system for proving depth boundedness.
    The term $P$ stands for $\new X.\protect\Parallel_{i \in I} A_i$.
  }
  \label{fig:typesys}
\end{figure*}

The type system presents several non-standard features.
First, it is defined on normal forms as opposed to general \piterm{s}.
This choice is motivated by the fact that
different syntactic presentations of the same term
may be misleading when trying to analyse the relation
between the structure of the term and~$\Types$.
The rules need to guarantee that a reduction will not break \tcompat[ibility],
which is a property of the congruence class of the term.
As justified by Lemma~\ref{lemma:tied-tree},
the scope of names in a congruence class may vary,
but the tied-to relation puts constraints on the structure
that must be obeyed by all members of the class.
Therefore the type system is designed around this basic concept,
rather than the specific scoping of any representative
of the structural congruence class.
Second, no type information is associated with the typed term,
only restricted names hold type annotations.
Third, while the rules are compositional,
the constraints on base types have a global flavour
due to the fact that they involve the structure of $\Types$
which is a global parameter of typing proofs.

\iflongversion
  Let us illustrate intuitively how the constraints enforced by the rules
  guarantee preservation of \tcompat[ibility].
  Consider the term
  \[
    P = \new e\:a. \Bigl(
      \new b.\bigl( \out a<b>.A_0 \bigr)
        \parallel
      \new d.\bigl( \inp a(x).A   \bigr)
    \Bigr)
  \]
  with
  $A = \new c. (
          A_1 \parallel A_2 \parallel A_3
        )$,
  $A_0 = \inp b(y)$,
  $A_1 = \out x<c>$,
  $A_2 = \inp c(z).\out a<e>$ and
  $A_3 = \out a<d>$.
  Let $\Types$ be the forest with
    $t_e \parent t_a \parent t_b \parent t_c$ and $t_a \parent t_d$,
  where $t_x$ is the base type of the (omitted) annotation
  of the restriction $\restr x$, for $x \in \set{a,b,c,d,e}$.
  The reader can check that $\forest(P)$ is \tcompat.
  In the traditional understanding of mobility,
  we would interpret the communication of $b$ over $x$
  as an application of scope extrusion
  to include $\new d.\bigl( \inp a(x).A \bigr)$ in the scope of $b$
  and then syncronisation over $a$ with the application of the substitution
  $\subst{x->b}$ to $A$;
  note that the substitution is only valid because the scope of $b$
  has been extended to include the receiver.
  Our key observation is that we can instead interpret this communication
  as a migration of the subcomponents of $A$ that do get their scopes
  changed by the reduction, from the scope of the receiver
  to the scope of the sender.
  For this operation to be sound,
  the subcomponents of $A$ migrating to the sender's scope
  cannot use the names that are in the scope of the receiver but not of the sender.
  In our specific example,
  after the synchronisation between the prefixes $\out a<b>$ and $\inp a(x)$,
  $b$ is substituted to $x$ in $A_1$ resulting in the term $A_1' = \out b<c>$
  and $A_0, A_1', A_2$ and $A_3$ become active.
  The scope of $A_0$ can remain unchanged
  as it cannot know more names than before as a result of the communication.
  By contrast, $A_1$ now knows $b$ as a result of the substitution $\subst{x->b}$:
  $A_1$ needs to migrate under the scope of $b$.
  Since $A_1$ uses $c$ as well, the scope of $c$ needs to be moved under $b$;
  however $A_2$ uses $c$ so it needs to migrate under $b$ with the scope of $c$.
  $A_3$ instead does not use neither $b$ nor $c$ so it can avoid migration
  and its scope remains unaltered.
  This information can be formalised using the tied-to relation:
  on one hand, $A_1$ and $A_2$ need to be moved together
  because $1 \tiedto{A} 2$ and
  they need to be moved because $x \ntiedto{a(x).A} 1, 2$.
  On the other hand, $A_3$ is not tied to neither $A_1$ nor $A_2$ in $A$ and
  does not know $x$, thus it is not migratable.
  After reduction, our view of the reactum is the term
  \[
    \new a. \Bigl(
      \new b.\bigl(
        A_0 \parallel
        \new c. ( A_1' \parallel A_2 )
      \bigr)
    \parallel \new d.A_3
    \Bigr)
  \]
  the forest of which is \tcompat.
  Rule~\ref{rule:conf}, applied to $A_1$ and $A_2$,
  ensures that $c$ has a base type that can be nested under the one of $b$.
  Rule~\ref{rule:in} does not impose constraints on the base types of $A_3$
  because $A_3$ is not migratable.
  It does however check that the base type of $e$ is an ancestor of the one of $a$,
  thus ensuring that both receiver and sender are already in the scope of $e$.
  The base type of $a$ does not need to be further constrained since
  the fact that the synchronisation happened on it implies that
  both the receiver and the sender were already under its scope;
  this implies, by \tcompat[ibility] of $P$, that $c$ can be nested under $a$.

  We now describe the purpose of the rules of the type system in more detail.
\fi
Most of the rules just drive the derivation through the structure of the term.
The crucial constraints are checked by~\ref{rule:conf}, \ref{rule:in} and \ref{rule:out}.

The \ref{rule:out}~rule is the one enforcing types
to be consistent with the dataflow of the process:
the type of the argument of a channel $a$ must agree with the types
of all the names that may be sent over $a$.
This is a very coarse sound over-approximation of the dataflow;
if necessary it could be refined using well-known techniques from the literature
but a simple approach is sufficient here to type interesting processes.

Rule~\ref{rule:conf} is best understood imagining the normal form to be typed, $P$, as the continuation of a prefix $\pi.P$.
In this context a reduction exposes each of the active sequential subterms of $P$ which need to have a place in a \tcompat\ forest for the reactum.
The constraint in \ref{rule:conf} can be read as follows.
A ``new'' leaf $A_i$ may refer to names already present in the forests of the reaction context; these names are the ones mentioned in both $\freenames(A_i)$ and $\Env$.
Then we must be able to insert $A_i$ so that we can find these names in its path. However, $A_i$ must belong to a tree containing all the names in $X$ that are tied to it in $P$. So by requiring every name tied to $A_i$ to have a base type smaller than any name in the context that $A_i$ may refer to, we make sure that we can insert the continuation in the forest of the context without violating \tcompat[ibility].
Note that $\Env(\freenames(A_i))$ contains only types that annotate names both in $\Env$ and $\freenames(A_i)$, that is, names which are not restricted by $X$ and are referenced by $A_i$ (and therefore come from the context).

Rule~\ref{rule:in} serves two purposes:
on the one hand it requires the type of the messages that can be sent through $a$ to be consistent with the use of the variable $x$ which will be bound to the messages;
on the other hand, it constrains the base types of $a$ and $x$ so that
synchronisation can be performed without breaking \tcompat[ibility].
The second purpose is achieved by distinguishing two cases,
represented by the two disjuncts of the condition on base types of the rule.
In the first case the base type of the message is an ancestor of the base type of $a$ in $\Types$.
This implies that in any \tcompat\ forest representing $a(x).P$,
the name $b$ sent as message over $a$ is already in the scope of $P$.
Under this circumstance, there is no real migration and the substitution
$\subst{x->b}$ does not alter the scope of $P$ and the \tcompat[ibility] constraints to be satisfied are in essence unaltered.
The second case is more complicated as it involves migration.
This case also requires a slightly non-standard feature:
the premises predicate not only on the direct subcomponents of an input prefixed term, but also on the direct subcomponents of the continuation.
This is needed to be able to separate the continuation in two parts:
the one requiring migration and the one that does not.
The non migratable sequential terms behave exactly as the case of the first disjunct: their scope is unaltered.
The migratable ones instead are intended to be inserted as descendent of the node representing the message in the forest of the reaction context.
For this to be valid without rearrangement of the forest of the context,
we need all the names in the context that are referenced in the migratable terms, to be already in their scope; we make sure this is the case by requiring the free names of any migratable $A_i$ that are from the context (i.e.~in $\Env$) to have base types smaller than the base type of $a$.
The set $\base(\Env(\freenames(A_i)\setminus\set{a}))$ indeed represents the base types of the names in the reaction context referenced in a migratable continuation $A_i$.
In fact $a$ is a name that needs to be in the scope of both the sender and the receiver at the same time, so it needs to be a common ancestor of sender and receiver in any \tcompat\ forest. Any name in the reaction context and in the continuation of the receiver, with a base type smaller than the one of $a$, will be an ancestor of $a$---and hence of the sender, the receiver and the node representing the message---in any \tcompat\ forest. Clearly, remembering $a$ is not harmful as it must be already in the scope of receiver and sender so we exclude it from the constraint.

\begin{example}
  Take the normal form in Example~\ref{ex:servers}.
  Let us fix $\Types$ to be the forest
    $\ty{s} \parent \ty{c} \parent \ty{m} \parent \ty{d}$
  and annotate the normal form with the following types:
    $s \tas \type_s = \ty{s}[\type_m]$,
    $c \tas \type_c = \ty{c}[\type_m]$,
    $m \tas \type_m = \ty{m}[\ty{d}] $ and
    $d \tas \ty{d}$.
  Let $\Env = \set{(s \tas \type_s), (c \tas \type_c)}$.
  We want to prove $\emptyset \types \new s\:c.P$.
  We can apply rule~\ref{rule:conf}:
  in this case there are no conditions on types because,
  being the environment empty,
  we have $\base(\emptyset(\freenames(A))) = \emptyset$
  for every active sequential term $A$ of $P$.
  The rule requires
    $\Env\types \bang{S}$, $\Env\types \bang{C}$ and $\Env\types \bang{M}$,
  which can be proved by proving typability of $S$, $C$ and $M$ under $\Env$
  by rule~\ref{rule:bang}.
  To prove $\Env\types S$ we apply rule~\ref{rule:in};
  we have $s \tas \ty{s}[\type_m] \in \Env$
  and we need to prove that $\Env, x\tas\type_m \types \new d.\out x<d>$.
  No constraints on base types are generated at this step
  since the migratable sequential term $\new d.\out x<d>$
  does not contain free variables typed by $\Env$ making
  $\Env(\freenames(\new d.\out x<d>) \setminus \set{a}) = \Env(\set{x})$ empty.
  Next, $\Env, x\tas\type_m \types \new d.\out x<d>$ can be proved by
  applying rule~\ref{rule:conf} which amounts to
  checking $\Env, x\tas\type_m \types \out x<d>.\zero$
  (by a simple application of \ref{rule:out}
  and the axiom $\Env, x\tas\type_m \types \zero$)
  and verifying the condition---true in $\Types$---$\base(\type_m) < \base(\type_d)$:
  in fact $d$ is tied to $\out x<d>$ and,
  for $\Env' = \Env \union \set{x \tas \type_m}$,
  $\base(\Env'(\freenames(\out x<d>)))
    = \base(\Env'(\set{x,d}))
    = \base(\set{\type_m})$.
  The proof for $\Env \types M$ is similar and requires
  $\ty{c} \tlt \ty{m}$ which is true in $\Types$.
  Finally, we can proof $\Env \types C$ using rule~\ref{rule:in};
  both the two continuation $A_1 = \out s<m>$ and $A_2 = \inp m(y).\out c<m>$
  are migratable in $C$ and since $\base(\type_m) \tlt \base(\type_c)$ is false
  we need the other disjunct of the condition to be true.
  This amounts to check that
    $\base(\Env(\freenames(A_1) \setminus \set{c}))
      = \base(\Env(\set{s,m})) = \base(\set{\type_s}) < \ty{c}$
  (note $m \not\in \domain(\Env)$)
  and
    $\base(\Env(\freenames(A_a) \setminus \set{c}))
      = \base(\Env(\set{})) < \ty{c}$
  (that holds trivially).
  Fortunately, this is the case in $\Types$.
  To complete the typing we need to show
  $\Env, m \tas \type_m \types A_1$ and
  $\Env, m \tas \type_m \types A_2$.
  The former can be proved by a simple application of \ref{rule:out}
  which does not impose further constraints on $\Types$.
  The latter is proved by applying \ref{rule:in} which requires
  $\base(\type_c) < \ty{m}$, which holds in $\Types$.
  Note how, at every step, there is only one rule that applies to each subproof.
\end{example}

\begin{example}
  There is no choice for (a finite) $\Types$ that would make
  the normal form in Example~\ref{ex:stack} typeable.
  To see why, one can build the proof tree without assumptions on $\Types$
  obtaining that:
  \begin{enumerate}
    \item the restrictions must be annotated with types consistent with the type assignments
      \ifshortversion%
        $s \tas t_s[t]$,
        $v \tas t_v[t]$,
        $n \tas t_n[t]$,
        $a \tas t$,
        $b \tas t$.
      \else%
        \begin{align*}
          s &\tas t_s[t] &
          v &\tas t_v[t] &
          n &\tas t_n[t] &
          a &\tas t      &
          b &\tas t
        \end{align*}
      \fi%
    \item $\Types$ must satisfy the constraint
          that the base type assigned to $b$ must be
          strictly greater than the one assigned to $x$,
          which is inconsistent with $s \tas t_s[t], b \tas t$.
  \end{enumerate}
\end{example}

\subsection{Soundness}
\label{sec:soundness}

In this section we show how the type system can be used
to prove depth-boundedness.
Theorem~\ref{th:subj-red} will show how typability is preserved by reduction.
Theorem~\ref{th:typed-tshaped} establishes the main property of the type system:
if a term is typable then \tshaped[ness] is invariant under reduction.
This allows us to conclude that if a term is \tshaped{} and typable,
then every term reachable from it will be \tshaped{}
and, therefore, it is depth-bounded.

We start with some simple properties of the type system.

\begin{lemma}
\label{lemma:typesys-props}
  Let $P \in \PiNfAnnot$ and $\Env$, $\Env'$ be type environments.
  \begin{enumerate}[label=\alph*), ref={\ref{lemma:typesys-props}.\alph*}]
    \item if $\Env \types P$ then $\freenames(P) \subseteq \domain(\Env)$;
      \label{lemma:typesys-domain}
    \item if  $\domain(\Env') \inters \boundnames(P) = \emptyset$
          and $\freenames(P) \subseteq \domain(\Env)$,
          then $\Env \types P$ if and only if $\Env \Env' \types P$;
      \label{lemma:typesys-weakening}
    \item if $P \congr P' \in \PiNfAnnot$ then, $\Env \types P$ if and only if $\Env \types P'$.
      \label{lemma:typesys-congr}
  \end{enumerate}
\end{lemma}

The subtitution lemma states that substituting names
without altering the types preserves typability.

\begin{lemma}[Substitution]
\label{lemma:subst}
  Let $P \in \PiNfAnnot$ and
  $\Env$ be a typing environment including the type assignments
  $a \tas \type$ and $b \tas \type$.
  Then it holds that if $\Env \types P$ then $\Env \types P\subst{a->b}$.
\end{lemma}

\iflongversion
\begin{proof}
  \input{proofs/substitution}
\end{proof}
\fi

Before we state the main theorem, we define the notion of
$P$-safe type environment, which is a simple restriction on the types
that can be assigned to names that are free at the top-level of a term.

\begin{definition}
A type environment $\Env$ is said to be \emph{$P$-safe} if
for each $x \in \freenames(P)$ and $(y\tas\type) \in \resboundnames(P)$,
$\base(\Env(x)) < \base(\type)$.
\end{definition}

\begin{theorem}[Subject Reduction]
\label{th:subj-red}
  Let $P$ and $Q$ be two terms in $\PiNfAnnot$
  and $\Env$ be a $P$-safe type environment.
  If $\Env \types P$ and $P \redto Q$, then $\Env \types Q$.
\end{theorem}

\begin{proof}
  \input{proofs/subj-red}
\end{proof}

\begin{theorem}[Invariance of {\tshaped[ness]}]
\label{th:typed-tshaped}
  Let $P$ and $Q$ be terms in $\PiNfAnnot$ such that $P \redto Q$ and
  $\Env$ be a $P$-safe environment such that $\Env \types P$.
  Then, if $P$ is \tshaped{} then $Q$ is \tshaped{}.
\end{theorem}

\begin{proof}
  \ifshortversion
    \input{proofs/typed-tshaped-sketch}
  \else
    \input{proofs/typed-tshaped}
  \fi
\end{proof}

\iflongversion
  To illustrate the role of $\phimig$, $\phinonmig$
  and the $\treeins$ operation in the above proof,
  we show an example that would not be typable if we choose a simpler
  ``migration'' transformation.
  \begin{example}
    \input{examples/phimig}
  \end{example}
\fi

\begin{definition}[Typably Hierarchical term]
\label{def:hierarchical}
  A normal form $P$ is \emph{typably hierarchical} if
  $P$ is \tshaped\ and $\Env \types P$
  for some finite forest $\Types$ and $P$-safe environment $\Env$.
  A general \piterm\ $P$ is typably hierarchical if its normal form $\nf(P)$ is.
\end{definition}

\begin{theorem}[Depth-boundedness]
\label{th:hier-db}
  Every typably hierarchical term is depth-bounded.
\end{theorem}

\iflongversion
\begin{proof}
  By Theorem~\ref{th:subj-red} and Theorem~\ref{th:typed-tshaped}
  every term reachable from a typably hierarchical term $P$ is \tshaped.
  Then by Proposition~\ref{prop:t-compat-db} $P$ is depth-bounded.
\end{proof}
\fi

\subsection{Type inference}
\label{sec:inference}

In this section we will show that it is possible to take
any non-annotated normal form $P$
and derive a forest $\Types$ and an annotated normal form for $P$
that can be typed under $\Types$.

It is straightforward to see that inference is decidable:
if a forest of base types can be found so that
the typing derivation for $P$ is successful,
there exists a $\Types$ with at most $\card{\boundnames(P)}$ nodes
and a $P$-safe environment $\Env$ with $\domain(\Env) = \freenames(P)$,
such that $\Env \types P$ and $P$ is \tshaped.
Therefore, a na\"ive algorithm could enumerate all such forests---there are finitely many---and type check $P$ against each.
However a better algorithm is possible.

We start by annotating the term with type variables:
each name $x$ gets typed with a type variable $\typevar_x$.
Then we start the type derivation,
collecting all the constraints on types along the way.
If we can find a $\Types$ and type expressions
to associate to each type variable, so that these constraints are satisfied,
the process can be typed under $\Types$.

The constraints have two forms:
\begin{enumerate}
  \item $\typevar_x = t_x[\typevar_y]$ where $t_x$ is a base type variable;
  \item $\base(\typevar_x) < \base(\typevar_y)$
        which correspond to constraints over the corresponding
        base type variables, i.e.~$t_x < t_y$.
\end{enumerate}
Note that the $P$-safe condition on $\Env$ translates to constraints of the second kind.
The first kind of constraints can be solved using unification.
If no solution exists, the process cannot be typed.
This is the case of processes that cannot be \emph{simply typed}~\cite{PierceS00}.
If unification is successful we get a set of equations where the unknowns
are the base type variables.
Any assignment of those variables to nodes in a suitable forest
that satisfies the constraints of the second kind
would be a witness of typability.

We have at most $n$ base type variables where $n$ is
the number of names occurring in $P$.
There are at most $\frac{n(n-1)}{2}$ distinct independent constraints
of the form $t_x < t_y$, which can be treated as uninterpreted propositions.
By inspecting rules~\ref{rule:conf} and \ref{rule:in} we observe that
all the ``tied-to'' and ``migratable'' predicates do not depend on $\Types$
so for any given $P$, the conjunction of constraints on base types generated in
the proof derivation forms a 2-CNF formula with $\bigO{n^2}$ boolean variables.
Since 2-CNF satisfiability is linear in the number of variables~\cite{2-cnf},
we obtain a $\bigO{n^2}$ bound on satisfiability of the base type constraints.
Once we prove satisfiability of these constraints,
to prove $P$ is typably hierarchical,
it remains to prove that there exists a model $\Types$ of the constraints
so that $P$ is \tshaped.
If a precise bound on the depth is needed, one can perform a search for the
shallowest forest which is a model of the base type constraints such that $P$ is \tshaped.
Otherwise, the search can be restricted to total orders.

%% file: definitions/typesys.tex
\begin{mathpar}
\inferrule*[right=Par]{
  \label{rule:conf}
    \forall i \in I \st
      \Env, X \types A_i
    \\
    \forall i \in I \st
      \forall x \tas \type_x \in X \st
        x \ntiedto{P} i
          \implies
            \base(\Env(\freenames(A_i))) \tlt \base(\type_x)
}{
    \Env \types \new X.\Parallel_{i\in I} A_i
}
\and 
\inferrule*[right=Choice]{
  \label{rule:sum}
    \forall i \in I \st \Env \types \pi_i.P_i
}{
    \Env \types \Alt_{i\in I} \pi_i.P_i
}
\and 
\inferrule*[right=Repl]{
  \label{rule:bang}
    \Env \types A
}{
    \Env \types \bang A
}
\and 
\inferrule*[right=Tau]{
  \label{rule:tau}
    \Env \types A
}{
    \Env \types \tact.A
}
\and 
\inferrule*[right=Out]{
  \label{rule:out}
    a \tas t_a[\type_b] \in \Env
    \\
    b \tas \type_b \in \Env
    \\
    \Env \types Q
}{
    \Env \types \out a<b>.Q
}
\and 
\inferrule*[right=In]{
  \label{rule:in}
    a \tas t_a[\type_x] \in \Env
    \\
    \Env, x \tas \type_x \types P
    \\
    \base(\type_x) \tleq t_a
    \lor \bigl(
    \forall i \in I \st
      \migr{a(x).P}(i) \implies
        \base(\Env(\freenames(A_i)\setminus\set{a})) \tlt t_a
    \bigr)
}{
    \Env \types \inp a(x).\new X.\Parallel_{i \in I} A_i
}
\end{mathpar}

%% file: proofs/substitution.tex
%
We prove the lemma by induction on the structure of $P$.
The base case is when $P \congr \zero$, where the claim trivially holds.

For the induction step,
let $P \congr \new X. \Parallel_{i \in I} A_i$
with $A_i = \Alt_{j \in J} \pi_{ij}.P_{ij}$,
for some finite sets of indexes $I$ and $J$.
Since the presence of replication does not affect the typing proof,
we can safely ignore that case as it follows the same argument.
Let us assume $\Env\types P$ and prove that $\Env\types P\subst{a->b}$.

Let $\Env'$ be $\Env \union X$.
From $\Env\types P$ we have
\begin{gather}
  \Env, X \types A_i
  \label{eq:subst-Ai}\\
      x \ntiedto{P} i
        \implies
          \base(\Env(\freenames(A_i))) \tlt \base(\type_x)
  \label{eq:subst-tied}
\end{gather}
for each $i \in I$ and $x \tas \type_x \in X$.
To extract from this assumptions a proof for $\Env\types P\subst{a->b}$,
we need to prove that \eqref{eq:subst-Ai} and \eqref{eq:subst-tied}
hold after the substitution.

Since the substitution does not apply to names in $X$
and the \emph{tied to} relation is only concerned
with names in $X$, the only relevant effect of the substitution is
modifying the set $\freenames(A_i)$ to
  $\freenames(A_i \subst{a->b}) =
      \freenames(A_i) \setminus \set{a} \union \set{b}$
when $a \in \freenames(A_i)$;
But since $\Env(a) = \Env(b)$ by hypothesis,
we have $\base(\Env(\freenames(A_i \subst{a->b}))) \tlt \base(\type_x)$.

It remains to prove \eqref{eq:subst-Ai} holds after the substitution as well.
This amounts to prove for each $j \in J$ that
$
    \Env' \types \pi_{ij}.P_{ij}
        \implies
    \Env' \types \pi_{ij}.P_{ij}\subst{a->b}
$;
we prove this by cases.

Suppose $\pi_{ij} = \out{\alpha}<\beta>$ for two names $\alpha$ and $\beta$,
then from $\Env' \types \pi_{ij}.P_{ij}$ we know the following
\begin{gather}
  \alpha \tas t_\alpha[\type_\beta] \in \Env'
  \qquad
  \beta \tas \type_\beta \in \Env'
  \label{eq:subst-out-types}\\
  \Env' \types P_{ij}
  \label{eq:subst-out-induction}
\end{gather}
Condition~\eqref{eq:subst-out-types} is preserved after the substitution
because it involves only types so, even if $\alpha$ or $\beta$ are $a$,
their types will be left untouched after they get substituted with $b$
from the hypothesis that $\Env(a) = \Env(b)$.
Condition~\eqref{eq:subst-out-induction} implies
$\Env' \types P_{ij}\subst{a->b}$ by inductive hypothesis.

Suppose now that $\pi_{ij} = \inp{\alpha}(x)$
and $P_{ij} \congr \new Y.\Parallel_{k \in K} A'_k$
for some finite set of indexes $K$;
by hypothesis we have:
\begin{gather}
  \alpha \tas t_{\alpha}[\type_x] \in \Env'
  \label{eq:subst-in-typea} \\
  \Env', x \tas \type_x \types P_{ij}
  \label{eq:subst-in-induction} \\
  \begin{multlined}
  \base(\type_x) \tleq t_{\alpha}
  \lor \\
  \forall k \in K \st
    \migr{\pi_{ij}.P_{ij}}(k) \implies
      \base(\Env'(\freenames(A'_k)\setminus\set{\alpha})) \tlt t_{\alpha}
  \end{multlined}
  \label{eq:subst-in-migrate}
\end{gather}
Now $x$ and $Y$ are bound names so they are not altered by substitutions.
The substitution $\subst{a->b}$ can therefore only be affecting
the truth of these conditions when
$\alpha =  a$ or when $a \in \freenames(A'_k)\setminus (Y \union \set{x})$.
Since we know $a$ and $b$ are assigned the same type by $\Env$
and $\Env \subseteq \Env'$,
condition~\eqref{eq:subst-in-typea} still holds when substituting $a$ for $b$.
Condition~\eqref{eq:subst-in-induction} holds by inductive hypotesis.
The first disjunct of condition~\eqref{eq:subst-in-migrate}
depends only on types, which are not changed by the substitution,
so it holds after applying it if and only if it holds before the application.
To see that the second disjunct also holds after the substitution
we observe that the \emph{migratable} condition depends on
$x$ and $\freenames(A'_k) \inters Y$
which are preserved by the substitution;
moreover, if $a \in \freenames(A'_k)\setminus\set{\alpha}$ then
$\Env'(\freenames(A'_k)\setminus\set{\alpha}) =
    \Env'(\freenames(A'_k\subst{a->b})\setminus\set{\alpha})$.

This shows that the premises needed to derive
$\Env', x \tas \type_x' \types \pi_{ij}.P_{ij}\subst{a->b}$
are implied by our hypothesis,
which completes the proof.

%% file: proofs/subj-red.tex
We will only prove the result for the case when $P \redto Q$ is caused by
a synchronising send and receive action
since the $\tact$ action case is similar and simpler.
From $P\redto Q$ we know that
$P \congr \new W.(S \parallel R \parallel C) \in \PiNfAnnot$ with
$S \congr (\out a<b>.\new Y_s.S')+M_s$ and
$R \congr (\inp a(x).\new Y_r.R')+M_r$
the synchronising sender and receiver respectively;
$Q \congr \new W Y_s Y_r.
  (S' \parallel R'\subst{x->b} \parallel C)$.
In what follows, let $W' = W Y_s Y_r$,
  $C = \Parallel_{h\in H} C_h$,
  $S' = \Parallel_{i\in I} S'_i$ and
  $R' = \Parallel_{j\in J} R'_j$,
all normal forms.

For annotated terms, the type system is syntax directed:
there can be only one proof derivation for each typable term.
By Lemma~\ref{lemma:typesys-congr}, from the hypothesis $\Env \types P$
we can deduce $\Env \types \new W.(S \parallel R \parallel C)$.
The proof derivation for this typing judgment
can only be of the following shape:
\begin{equation}
\label{eq:P-deriv}
  \mprset{sep=1em}
  \inferrule{
    \Env W \types S \\
    \Env W \types R \\
    \forall h \in H \st
      \Env W \types C_h \\
    \Premise
  }{
    \Env \types \new W.(S \parallel R \parallel C)
  }
\end{equation}
where $\Premise$ represents the rest of the conditions of the \ref{rule:conf} rule.%
\footnote{Note that $\Premise$ is trivially true by $P$-safety of $\Env$.}
The fact that $P$ is typable implies that each of these premises must be provable.
The derivation proving $\Env, W \types S$ must be of the form
\begin{equation}
\label{eq:S-deriv}
  \mprset{sep=1em}
  \inferrule{
    \inferrule*{
      a \tas t_a[\type_b] \in \Env W \\
      b \tas \type_b \in \Env W \\
      \Env W \types \new Y_s.S'
    }{
      \Env W \types \out a<b>.\new Y_s.S'
    }\\
    \Premise_{M_s}\\
  }{
    \Env \types \out a<b>.\new Y_s.S' + M_s
  }
\end{equation}
where $\Env W \types \new Y_s.S'$ is proved by an inference of the shape
\begin{equation}
\label{eq:Si-deriv}
  \inferrule{
    \forall i\in I\st
      \Env W Y_s \types S'_i \\
    \forall i\in I\st
      \Premise_{S'_i}
  }{
    \Env W \types \new Y_s.S'
  }
\end{equation}

Analogously, $\Env W \types R$ must be proved by an inference with the following shape
\begin{equation}
\label{eq:R-deriv}
  \mprset{sep=1em}
  \inferrule{
    \inferrule*{
      a \tas t_a[\type_x] \in \Env W \\
      \Env W, x \tas \type_x \types \new Y_r.R'\\
      \Premise_{R'}
    }{
      \Env W \types \inp a(x).\new Y_r.R'
    }\\
    \Premise_{M_r}
  }{
    \Env W \types \inp a(x).\new Y_r.R' + M_r
  }
\end{equation}
and to prove $\Env W, x \tas \type_x \types \new Y_r.R'$
\begin{equation}
\label{eq:Rj-deriv}
  \inferrule{
    \forall j\in J\st
      \Env W, x \tas \type_x, Y_r \types R'_j\\
    \forall j\in J\st
      \Premise_{R'_j}
  }{
    \Env W, x \tas \type_x \types \new Y_r.R'
  }
\end{equation}

We have to show that from this hypothesis we can infer that $\Env \types Q$
or, equivalently (by Lemma~\ref{lemma:typesys-congr}), that $\Env \types Q'$
where $Q' = \new W Y_s Y_r. (S' \parallel R'\subst{x->b} \parallel C)$.
The derivation of this judgment can only end with an application of~\ref{rule:conf}:
\[\mprset{sep=1em}
  \inferrule{
    \forall i \in I \st
      \Env W' \types S'_i \\
    \forall j \in J \st
      \Env W' \types R'_j\subst{x->b} \\
    \forall h \in H \st
      \Env W' \types C_h \\
    \Premise'
  }{
    \Env \types \new W'. (S' \parallel R'\subst{x->b} \parallel C)
  }
\]
In what follows we show how we can infer these premises are provable
as a consequence of the provability of the premises of the proof of
$\Env \types \new W.(S \parallel R \parallel C)$.

From Lemma~\ref{lemma:typesys-weakening} and \ref{nameuniq},
$\Env W Y_s \types S'_i$ from~\eqref{eq:Si-deriv} implies
$\Env W' \types S'_i$ for each $i \in I$.

Let $\Env_r = \Env W, x \tas \type_x$.
We observe that by \eqref{eq:S-deriv} and \eqref{eq:R-deriv}, $\type_x = \type_b$.
From \eqref{eq:R-deriv} we know that $\Env_r Y_r \types R'_j$
which, by Lemma~\ref{lemma:subst}, implies $\Env_r Y_r \types R'_j\subst{x->b}$.
By Lemma~\ref{lemma:typesys-weakening} we can infer
$\Env_r Y_r Y_s \types R'_j\subst{x->b}$ and by applying the same lemma again
using $\freenames(R'_j\subst{x->b}) \subseteq \domain(\Env W Y_r Y_s)$
and \ref{nameuniq} we obtain
$\Env W' \types R'_j\subst{x->b}$.

Again applying Lemma~\ref{lemma:typesys-weakening} and \ref{nameuniq},
we have that $\Env W \types C_h$ implies $\Env W' \types C_h$ for each $h \in H$.

To complete the proof we only need to prove that for each
$A \in \set{S'_i | i \in I} \union \set{R'_j | j \in J} \union \set{C_h | h \in H}$,
$
  \Premise' = \forall (x \tas \type_x) \in W'\st
    x \text{ tied to } A \text{ in } Q' \implies
      \base(\Env(\freenames(A))) < \base(\type_x)
$
holds.
This is trivially true by the hypothesis that $\Env$ is $P$-safe.

%% file: proofs/typed-tshaped-sketch.tex
The full proof is presented in~\appendixorfull.
We outline here only the crucial steps.
Consider the input output synchronisation case where the sending action $\out a<b>$ is such that
$\restr (a \tas \type_a)$ and $\restr (b \tas \type_b)$ are both active restrictions of $P$ and $t_a=\base(\type_a)\tlt\base(\type_b)=t_b$.
As in the proof of Theorem~\ref{th:subj-red}, the derivation of $\Env \types P$
must follow the shape of \eqref{eq:P-deriv}.
From \tshaped[ness] of $P$ we can conclude that
both $\new Y_s.S'$ and $\new Y_r.R'$ are \tshaped.
By Lemma~\ref{lemma:phi-tcompat} we know that
  $\phi   = \Phi(\new W.(S \parallel R \parallel C)) \in \AST{P}$,
  $\phi_r = \Phi(\new Y_r.R'\subst{a->b}) \in \AST{\new Y_r.R'\subst{x->b}}$ and
  $\phi_s = \Phi(\new Y_s.S') \in \AST{\new Y_s.S'}$.
Let $\phi_r = \phimig \dunion \phinonmig$
where only $\phimig$ contains a leaf labelled with a term with $b$ as a free name.
These leaves will correspond to the continuations $R'_j$
that migrate in $\inp a(x).\new Y_r.R'$,
after the application of the substitution $\subst{x->b}$.
By assumption, inside $P$ both $S$ and $R$ are in the scope of the restriction bounding $a$
and $S$ must also be in the scope of the restriction bounding $b$.
$\phi$ will contain two leaves $n_S$ and $n_R$
labelled with $S$ and $R$ respectively,
having a common ancestor $n_a$ labelled with $(a, t_a)$;
$n_S$ will have an ancestor $n_b$ labelled with $(b, t_b)$.
Let $p_a$, $p_S$ and $p_R$ be the paths in $\phi$ leading from a root to $n_a$, $n_S$ and $n_R$ respectively.
\newcommand{\parentins}{\parent_{\mathrm{ins}}}
Now, we want to transform $\phi$, by manipulating this tree,
into a forest $\phi'$ that is
\tcompat\ by construction and such that
there exists a term $Q' \congr Q$ with $\forest(Q') = \phi'$,
so that we can conclude $Q$ is \tshaped.

To do so, we introduce the following function,
taking a labelled forest $\phi$, a path $p$ in $\phi$ and a labelled forest $\rho$
and returning a labelled forest:
$  \treeins(\phi, p, \rho) \is (
    N_\phi \dunion N_\rho,
    \parent_\phi \dunion \parent_\rho \dunion \parentins,
    \ell_\phi \dunion \ell_\rho
  )$
where $n\parentins n'$ if
  $n' \in \min_{\parent_\rho}(N_\rho)$,
  $\ell_\rho(n') = (y, t_y)$
and $n \in \max_{\parent_\phi}\set{m \in p | \ell_\phi(m) = (x,t_x), t_x < t_y}$.
Note that for each $n'$, since $p$ is a path, there can be at most one $n$ such that $n \parentins n'$.
To obtain the desired $\phi'$,
we first need to remove the leaves $n_S$ and $n_R$ from $\phi$,
as they represent the sequential processes which reacted,
obtaining a forest $\phi_C$.
We argue that the $\phi'$ we need is indeed
  $\phi' = \treeins(\phi_1, p_S, \phimig)$ where
  $\phi_1 = \treeins(\phi_2, p_R, \phinonmig)$ and
  $\phi_2 = \treeins(\phi_C, p_S, \phi_s)$.
It is easy to see that, by definition of $\treeins$, $\phi'$ is \tcompat:
$\phi_C$, $\phi_s$, $\phinonmig$ and $\phimig$ are \tcompat\ by hypothesis,
$\treeins$ adds parent-edges only when they do not break \tcompat[ibility].

The proof is completed by showing that $\phi'$ is the forest of a term congruent to
  $\new W Y_s Y_r.(S' \parallel R'\subst{x->b} \parallel C)$
using Lemma~\ref{lemma:forest-nf}.

%% file: proofs/typed-tshaped.tex
We will consider the input output synchronisation case as the $\tau$ action one is similar and simpler.
We will further assume that the sending action $\out a<b>$ is such that
$\restr (a \tas \type_a)$ and $\restr (b \tas \type_b)$ are both active restrictions of $P$,
i.e.~$(a \tas \type_a) \in W$, $(b \tas \type_b) \in W$
     with $P \congr \new W.(S \parallel R \parallel C)$.
The case when any of these two names is a free name of $P$ can be easily handled
with the aid of the assumption that $\Env$ is $P$-safe.

As in the proof of Theorem~\ref{th:subj-red}, the derivation of $\Env \types P$
must follow the shape of \eqref{eq:P-deriv}.

From \tshaped[ness] of $P$ we can conclude that
both $\new Y_s.S'$ and $\new Y_r.R'$ are \tshaped.
We note that substitutions do not affect \tcompat[ibility]
since they do not alter the set of bound names and their type annotations.
Therefore, we can infer that $\new Y_r.R'\subst{a->b}$ is \tshaped.
By Lemma~\ref{lemma:phi-tcompat} we know that
  $\phi   = \Phi(\new W.(S \parallel R \parallel C)) \in \AST{P}$,
  $\phi_r = \Phi(\new Y_r.R'\subst{a->b}) \in \AST{\new Y_r.R'\subst{x->b}}$ and
  $\phi_s = \Phi(\new Y_s.S') \in \AST{\new Y_s.S'}$.
Let $\phi_r = \phimig \dunion \phinonmig$
where only $\phimig$ contains a leaf labelled with a term with $b$ as a free name.
These leaves will correspond to the continuations $R'_j$
that migrate in $\inp a(x).\new Y_r.R'$,
after the application of the substitution $\subst{x->b}$.
By assumption, inside $P$ both $S$ and $R$ are in the scope of the restriction bounding $a$
and $S$ must also be in the scope of the restriction bounding $b$.
Let $t_a = \base(\type_a)$ and $t_b = \base(\type_b)$,
$\phi$ will contain two leaves $n_S$ and $n_R$
labelled with $S$ and $R$ respectively,
having a common ancestor $n_a$ labelled with $(a, t_a)$;
$n_S$ will have an ancestor $n_b$ labelled with $(b, t_b)$.
Let $p_a$, $p_S$ and $p_R$ be the paths in $\phi$ leading from a root to $n_a$, $n_S$ and $n_R$ respectively.
By \tcompat[ibility] of $\phi$, we are left with only two possible cases:
either \begin{inparaenum}[label=\arabic*)]
\item $t_a \tlt t_b$\label{case:ab} or \item $t_b \tlt t_a$.\label{case:ba}
\end{inparaenum}

Let us consider case~\ref{case:ab} first.
The tree in $\phi$ to which the nodes $n_S$ and $n_R$ belong,
would have the following shape:
\begin{center}%
  \ifshortversion\unskip\fi
  \tikzsetnextfilename{tree-before-sync}%
  \begin{tikzpicture}[y=.5cm,x=.6cm,baseline=0]
    \input{figures/tree-before-sync.tikz}
  \end{tikzpicture}
\end{center}
\ifshortversion\unskip\fi
\newcommand{\parentins}{\parent_{\mathrm{ins}}}
Now, we want to transform $\phi$, by manipulating this tree,
into a forest $\phi'$ that is
\tcompat\ by construction and such that
there exists a term $Q' \congr Q$ with $\forest(Q') = \phi'$,
so that we can conclude $Q$ is \tshaped.

To do so, we introduce the following function,
taking a labelled forest $\phi$, a path $p$ in $\phi$ and a labelled forest $\rho$
and returning a labelled forest:
\[
  \treeins(\phi, p, \rho) \is (
    N_\phi \dunion N_\rho,
    \parent_\phi \dunion \parent_\rho \dunion \parentins,
    \ell_\phi \dunion \ell_\rho
  )
\]
where $n\parentins n'$ if
  $n' \in \min_{\parent_\rho}(N_\rho)$,
  $\ell_\rho(n') = (y, t_y)$
and $n \in \max_{\parent_\phi}\set{m \in p | \ell_\phi(m) = (x,t_x), t_x < t_y}$.
Note that for each $n'$, since $p$ is a path, there can be at most one $n$ such that $n \parentins n'$.

To obtain the desired $\phi'$,
we first need to remove the leaves $n_S$ and $n_R$ from $\phi$,
as they represent the sequential processes which reacted,
obtaining a forest $\phi_C$.
We argue that the $\phi'$ we need is indeed
\begin{align*}
  \phi' &= \treeins(\phi_1, p_S, \phimig)    \\
  \phi_1 &= \treeins(\phi_2, p_R, \phinonmig)\\
  \phi_2 &= \treeins(\phi_C, p_S, \phi_s)
\end{align*}
It is easy to see that, by definition of $\treeins$, $\phi'$ is \tcompat:
$\phi_C$, $\phi_s$, $\phinonmig$ and $\phimig$ are \tcompat\ by hypothesis,
$\treeins$ adds parent-edges only when they do not break \tcompat[ibility].

To prove the claim we need to show that $\phi'$ is the forest of a term congruent to
  $\new W Y_s Y_r.(S' \parallel R'\subst{x->b} \parallel C)$.
\newcommand{\Jnonmig}{J_{\neg\mathrm{mig}}}
\newcommand{\Jmig}{J_{\mathrm{mig}}}
Let $R' = \Parallel_{j \in J} R'_j$,
    $\Jmig = \set{j \in J | x \ntiedto{\new Y_r.R'} j}$,
    $\Jnonmig = J \setminus \Jmig$
and $Y'_r = \set{(x \tas \type) \in Y_r
                  | x \in \freenames(R'_j), j \in \Jnonmig}$.
We know that no $R'_j$ with $j \in \Jnonmig$ can contain $x$ as a free name
so $R'_j\subst{x->b} = R'_j$.
Now suppose we are able to prove that conditions
  \ref{lemma:forest-nf:seq-leaf},
  \ref{lemma:forest-nf:name-uniq} and
  \ref{lemma:forest-nf:scoping}
of Lemma~\ref{lemma:forest-nf} hold for
  $\phi_C$, $\phi_1$, $\phi_2$ and $\phi'$.
Then we could use Lemma~\ref{lemma:forest-nf} to prove
\begin{enumerate}[label=\alph*)]
  \item $\phi_C = \forest(Q_C)$, $Q_C \congr Q_{\phi_C} = \new W.C$,
  \item $\phi_2 = \forest(Q_2)$, $Q_2 \congr Q_{\phi_2} = \new W Y_s.(S' \parallel C)$,
  \item $\phi_1 = \forest(Q_1)$, $Q_1 \congr Q_{\phi_1} = \new W Y_s Y'_r.(S' \parallel \Parallel_{j \in \Jnonmig} R'_j \parallel C)$,
  \item $\phi' = \forest(Q')$, $Q' \congr Q_{\phi'} = \new W Y_s Y_r.(S' \parallel R'\subst{x->b} \parallel C) \congr Q$
\end{enumerate}
(it is straightforward to check that $\phi_C, \phi_2, \phi_1$ and $\phi'$ have the right sets of nodes and labels to give rise to the right terms).
We then proceed to check for each of the forests above that they satisfy
conditions
  \ref{lemma:forest-nf:seq-leaf},
  \ref{lemma:forest-nf:name-uniq} and
  \ref{lemma:forest-nf:scoping},
thus proving the theorem.

Condition~\ref{lemma:forest-nf:seq-leaf} requires that
only leafs are labelled with sequential processes,
condition that is easily satisfied by all of the above forests
since none of the operations involved in their definition alters this property
and the forests $\phi$, $\phi_s$ and $\phi_r$ satisfy it by construction.

Similarly, since $\new W.(S \parallel R \parallel C)$ is a normal form
it satisfies \ref{nameuniq}, \ref{lemma:forest-nf:name-uniq} is satisfied
as we never use the same name more than once.

Condition \ref{lemma:forest-nf:scoping} holds on $\phi$ and hence it holds on $\phi_C$
since the latter contains all the nodes of $\phi$ labelled with names.

Now consider $\phi_s$: in the proof of Theorem~\ref{th:subj-red} we established
that $\Env \types P$ implies that the premises $\Premise_{S'_i}$
from \eqref{eq:Si-deriv} hold, that is
$\base(\Env W(\freenames(S'_i))) \tlt \base(\type_x)$
holds for all $S'_i$ for $i \in I$ and all $(x \tas \type_x) \in Y_s$
such that $x \ntiedto{\new Y_s.S'} i$.
Since $\freenames(S'_i) \inters W \subseteq \freenames(S')$ we know that
every name $(w \tas \type_w) \in W$ such that $w \in \freenames(S'_i)$
will appear as a label $(w, \base(\type_w))$ of a node $n_w$ in $p_S$.
Therefore, by definition of $\treeins$,
we have that for each $n \in N_{\phi_C}$, $n_w <_{\phi_2} n$;
in other words, in $\phi_2$, every leaf in $N_{\phi_s}$ labelled with $S'_i$
is a descendent of a node labelled with $(w, \base(\type_w))$
for each $(w \tas \type_w) \in W$ with $w \in \freenames(S'_i)$.
This verifies condition \ref{lemma:forest-nf:scoping} on $\phi_2$.

Similarly, by \eqref{eq:Rj-deriv} the following premise must hold:
$ \base(\Env W(\freenames(R'_j))) \tlt \base(\type_x) $
for all $R'_j$ for $j \in J$ and all $(y \tas \type_y) \in Y_r$
such that $y \ntiedto{\new Y_r.R'} j$.
We can then apply the same argument we applied to $\phi_2$ to show
that condition \ref{lemma:forest-nf:scoping} holds on $\phi_1$.

From \eqref{eq:R-deriv} and the assumption $t_a \tlt t_b$,
we can conclude that the following premise must hold:
$\base(\Env W(\freenames(R'_j) \setminus \set{a})) \tlt t_a$
for each $j \in J$ such that $R'_j$ is migratable in $\inp a(x).\new Y_r.R'$,
i.e~$j \in \Jmig$.
From this we can conclude that for every name $(w \tas \type_w) \in W$
such that $w \in \freenames(R'_j\subst{x->b})$ with $j \in \Jmig$
there must be a node in $p_a$ (and hence in $p_S$) labelled with $(w, \base(\type_w))$.
Now, some of the leaves in $\phimig$ will be labelled with terms having $b$ as a free name;
we show that in fact every node in $\phimig$ labelled with a
$(y, t_y)$ is indeed such that $t_y \tlt t_b$.
From the proof of Theorem~\ref{th:subj-red} and Lemma~\ref{lemma:subst}
we know that from the hypothesis we can infer that
$\Env W \types \new Y_r.R'\subst{x->b}$ and hence that
for each $j \in \Jmig$ and each $(y \tas \type_y) \in Y_r$,
if $y$ is tied to $R'_j\subst{x->b}$ in $\new Y_r.R'\subst{x->b}$ then
$\base(\Env W(R'_j\subst{x->b})) \tlt \base(\type_y)$.
By Lemma~\ref{lemma:phi-tied} we know that every root of $\phimig$
is labelled with a name $(y, t_y)$ which is tied to each of the leaves in its tree.
Therefore each such $t_y$ satisfies $\base(\Env W(R'_j\subst{x->b})) < t_y$.
By construction, there exists at least one $j \in \Jmig$ such that
$x \in \freenames(R'_j)$ and consequently such that
$b \in \freenames(R'_j\subst{x->b})$.
From this and $b \in W$ we can conclude $t_b < t_y$ for $t_y$ labelling a root in $\phimig$.
We can then conclude that
$\set{n_b} = \max_{\parent_{\phi_2}}\set{m \in p_S | \ell_\phi(m) = (z,t_z), t_z < t_y}$
for each $t_y$ labelling a root of $\phimig$,
which means that each tree of $\phimig$ is placed as a subtree of $n_b$ in $\phi'$.
This verifies condition~\ref{lemma:forest-nf:scoping} for $\phi'$
completing the proof.

Pictorially, the tree containing $n_S$ and $n_R$ in $\phi$ is now transformed
in the following tree in $\phi'$:
\begin{center}
  \tikzsetnextfilename{tree-after-sync}%
  \begin{tikzpicture}[scale=.65,baseline=0]
    \input{figures/tree-after-sync.tikz}
  \end{tikzpicture}
\end{center}

Case~\ref{case:ba} --- where $t_b < t_a$ --- is simpler as the migrating
continuations can be treated just as the non-migrating ones.

%% file: figures/tree-before-sync.tikz
\ensuretikzpicturebegin
\begin{scope}[
  forest,
  thick,
]
      \path[tree] (-2.6,-3.5) -- (0,0) coordinate (root) -- (2.6,-3.5) --cycle;
      \path[closer] (root)
        ++(-90:1)    node[node=above left:{$n_a$}] (a) {}
         +(-130:1.3) node[node=above left:{$n_b$}] (b) {}
      ;
      \path
         (1.3,-3.5)  node[node=below:{$n_R$}] (R) {}
         (-1.3,-3.5) node[node=below:{$n_S$}] (S) {}
      ;
  \draw[thick] (root)
    ..controls +(-120:.3) and +(60:.4) ..
    (a)
    ..controls +(-170:.5)  and +(30:.7) ..
    (b)
    ..controls +(-150:.7) and +(40:.8) ..
    (S.north)
    (a)
    ..controls +(-10:.5)  and +(260:-.6) ..
    ++(-60:1.2)
    ..controls +(260:.8)  and +(50:.7) ..
    (R.north);
\end{scope}
\ensuretikzpictureend

%% file: figures/tree-after-sync.tikz
\ensuretikzpicturebegin
\begin{scope}[
  forest,
  thick,
  legend/.style={
    every node/.style={
      anchor=west,
      text width=4em,
      append after command={++(0,-1em)}
    },
    every pic/.style={
      shift={(-.4em,.6em)},
      y=2em, x=2em
    }
  },
]

  \path[tree] (-2.6,-3.5) -- (0,0) coordinate (root) -- (2.6,-3.5) --cycle;
  \path[closer] (root)
    ++(-90:1)    node[node=above left:{$n_a$}] (a) {}
     +(-130:1.3) node[node=above left:{$n_b$}] (b) {}
  ;
  \path
     (1.3,-3.5)  node[removed node=below:{$p_R$}] (R) {}
     (-1.3,-3.5) node[removed node=below:{$p_S$}] (S) {}
  ;
  \draw (root)
    ..controls +(-120:.3) and +(60:.4) ..
    (a)
    ..controls +(-170:.5) and +(30:.7) ..
      coordinate[pos=.55](pre-b)
    (b)
    ..controls +(-150:.7) and +(40:.8) ..
      coordinate[pos=.2](s')
    (S.north)
    (a)
    ..controls +(-10:.5)  and +(260:-.6) ..
      coordinate[pos=.8](r')
    ++(-60:1.2)
    ..controls +(260:.8)  and +(50:.7) ..
      coordinate[pos=.1](r'')
    (R.north);

  \draw (b) to[bend left] ++(.3,-.3 ) pic[migrating subtree]
        (b) to[bend left] ++(.4,-.35) pic[migrating subtree]
        (b) to[bend left] ++(.5,-.4 ) pic[migrating subtree];

  \draw (s') to[bend right] ++(-.4,-.4) pic[sender subtree]
     (pre-b) to[bend left ] ++( .3,-.3) pic[sender subtree]
        (r') to[bend left ] ++( .5,-.3) pic[receiver subtree]
       (r'') to[bend right] ++(-.3,-.5) pic[receiver subtree];

  \path(-2,0) pic[sender subtree];

  \path( 2,0) pic[receiver subtree];

  \path[legend] (5,  -2em)
    {pic [sender subtree   ]} node {$\in \phi_s    $}
    {pic [migrating subtree]} node {$\in \phimig   $}
    {pic [receiver subtree ]} node {$\in \phinonmig$}
  ;

\end{scope}
\ensuretikzpictureend

%% file: examples/phimig.tex
Consider the normal form $P = \new a\:b\:c.(\bang{A} \parallel \out a<c>)$
where
  $A = \inp a(x).\new d.(\out a<d> \parallel \out b<x>)$.
To make types consistent we need annotations satisfying
$a \tas t_a[t]$, $b \tas t_b[t]$, $c \tas t$ and $d \tas t$.
Any $\Types$ satisfying the constraints $t_b \tlt t_a \tlt t$ would allow us
to prove $\emptyset\types P$;
let then $\Types$ be the forest with $\ty{b} \parent \ty{a} \parent \ty{t}$
with $t_a = \ty{a}$, $t_b = \ty{b}$ and $t = \ty{t}$.
Let $P' = \new a\:b\:c\:d.(\bang{A} \parallel \out a<d> \parallel \out b<c>)$
be the (only) successor of $P$.
The following picture shows $\Phi(P)$ in the middle,
on the left a forest in $\AST{P'}$ extracted by just putting
the continuation of $A$ under the message,
on the right the forest obtained by using $\treeins$
on the non-migrating continuations of $A$:
\begin{center}
  \tikzsetnextfilename{example-treeins}%
   \begin{tikzpicture}[yscale=.7]
     \input{figures/example-treeins.tikz}
   \end{tikzpicture}
\end{center}
Clearly, the tree on the left is not \tcompat\ since $c$ and $d$
have the same base type $t$.
Instead, the tree on the right can be obtained because $\treeins$
inserts the non-migrating continuation as close to the root as possible.

%% file: figures/example-treeins.tikz
\ensuretikzpicturebegin
\begin{scope}[
  child anchor=north,
  every node/.style={
    outer sep=0pt,
    inner sep=2pt,
  },
  level distance=6mm,
  sibling distance=8mm,
]

\node (P) {$b$}
  child { node {$a$}
    child { node {$\bang{A}$}}
    child { node {$c$}
      child[level distance=8mm] { node {$\out a<c>$} }
    }
  };

\node (Pins) at (3,0) {$b$}
  child { node {$a$}
    child { node {$\bang{A}$}}
    child { node {$c$}
      child[level distance=8mm] { node {$\out b<c>$} }
    }
    child { node {$d$}
      child[level distance=8mm] { node {$\out a<d>$} }
    }
  };

\node (Pnaive) at (-3,0) {$b$}
  child { node {$a$}
    child { node {$\bang{A}$}}
    child { node {$c$}
      child { node {$d$}
        child { node {$\out b<c>$} }
        child { node {$\out a<d>$} }
      }
    }
  };

\path (P) --node[yshift=-15]{$\to$}   (Pins)
      (P) --node[yshift=-15]{$\gets$} (Pnaive);

\end{scope}
\ensuretikzpictureend

%% file: 5-encodings.tex
\section{Equivalence with \NDA}
\label{sec:automata}

After isolating a fragment of a process calculus,
an interesting question is \emph{can we find an automata based presentation of the same fragment?}
In this section we give an answer to this question
by relating the typably hierarchical fragment
to a class of automata on data-words recently defined in~\cite{conrad:14}:
\emph{Nested Data Class Memory Automata} (\NDA).

The original presentation of \NDA{s} sees them as language recognition devices:
they can recognise sets of data-words,
that is sequences of symbols in $\Sigma \times \dataset$ where
  $\Sigma$ is a finite alphabet and
  $\dataset$ is an infinite set of \emph{data values}.
Notably, (weak) \NDA{s} are more expressive than Petri nets,
while enjoying decidability of some verification problems.
While \emph{Class Memory Automata}~\cite{cma}
do not postulate any structure on $\dataset$,
\NDA{s} assume that it is equipped with
an infinitely branching, finite height forest structure.
We will make use of this forest structure
to represent \tcompat\ \piterm\ forests.

We are primarily interested in establishing a tight relation between the
transition systems of \NDA{s} and typably hierarchical terms.
Therefore we do not regard \NDA{s} as language recognisation devices
but simply as computational models.
For this reason, our definition ignores the language-related components
of the original definition of~\cite{conrad:14}:
there is no finite alphabet $\Sigma$, no accepting control states,
no accepting run.
While in the language-theoretic formulation at each step in a run
a letter and a data value must be read from the input string,
here a transition can fire simply if \emph{there exists}
a data value satisfying the transition's precondition.

\begin{definition}[\NDA~\cite{conrad:14}]
\label{def:ndcma}\label{def:nda}
We define a \emph{nested dataset} $(\dataset, \pred_{\dataset})$
to be a forest of infinitely many trees of level $\ell$
which is \emph{full} in the sense that
for each data value $d$ of level less than $\ell$,
there are infinitely many data values $d'$ whose parent is $d$.

A \emph{class memory function} is a function
  $f \from \dataset \to A \dunion \set{\fresh}$
such that $f(d) = \fresh$ for all but finitely many $d \in \dataset$;
$\fresh$ is a special symbol indicating a data value is fresh,
i.e.~has never been used before.

Fix a nested data set of level $\ell$.
A \emph{Nested Data CMA of level $\ell$} is a tuple
  $(\States, \delta, q_0, f_0)$
where
  $\States$ is a finite set of states,
  $q_0 \in \States$ is the initial control state,
  $f_0 \from \dataset \to \States_\fresh$ is the initial class memory function
  satisfying $f_0(\pred(d)) = \fresh \implies f_0(d) = \fresh$,
  and $\delta$ is the transition relation.
$\delta$ is given by a union $\delta = \Union_{i=1}^\ell \delta_i$
where each $\delta_i$ is a relation:
  $\delta_i \subseteq \States \times (\States_\fresh)^i \times \States \times \States^i$
and $\States_\fresh$ is defined as $\States \union \set{\fresh}$.
A configuration is a pair $(q, f)$ where $q \in \States$,
and $f \from \dataset \to \States_\fresh$ is a class memory function.
The initial configuration is $(q_0, f_0)$.
The automaton can transition from configuration
$(q, f)$ to configuration $(q',f')$, written $(q, f) \autto (q',f')$,
just if there exists a level-$i$ data value $d$ such that
$(q, \lstc{q}{i}, q', \lstc{q'}{i}) \in \delta$,
for all $j \in \set{1,\dots,i}$, $f(\pred^{i-j}(d)) = q_j$
and
\[
  f' = f \Map{ \pred^{i-1}(d) -> q'_1; \dots,\:\pred(d) -> q'_{i-1}; d -> q'_i}.
\]
\end{definition}
Given a nested dataset $\dataset$ we write $\CMF$ for the set of all
class memory functions from $\dataset$ to $\States_\fresh$.

We want to show that, in some strong sense,
\NDA{s} are equi-expressive to typably hierarchical \piterm{s}.
First we show an encoding from typeable \piterm{s},
then we prove that a transition system generated from the \NDA\ encoding
is bisimilar to the transition system generated by
the reduction semantics of the \piterm.
This result is quite strong in that it
implies the equivalence of many decision problems of the two formalisms.
It also offers a bridge between infinite-alphabet automata
and decidable fragments of \picalc.

\subsection{Encoding Typably Hierarchical terms into \NDA}
\label{sec:encoding}

We make a few simplifying assumptions on the term to be encoded as an \NDA.
First, we assume $P$ is a closed normal form, i.e.~$\freenames(P) = \emptyset$,
second we assume $P$ contains no $\tact$ action.
It would be easy to support the general case
but we only focus on the core case for conciseness.
Fix a closed \tshaped\ \piterm\ $P$ such that $\emptyset \types P$,
with $\ell = \height(\Types)$.
We will construct a level-$\ell$ automaton $\AutEnc{P}$ from $P$
so that their transition systems are essentially bisimilar.

The intuition behind the encoding is as follows.
A configuration $(q, f)$ represents a \piterm\ $P$ by using $f$
to label a finite portion of $\dataset$ so that
it is isomorphic to a \tcompat\ forest in $\AST{P}$.
Our encoding proceeds in rounds.
A single synchronisation step between two processes
will be simulated by a predictable number of steps of the automaton.
Since \piterm{s} exhibit non-determinism,
the automata in the image of the encoding need to be non-deterministic as well.
We make use of the non-determinism of the automata model in a second way:
in a reduction, the two synchronising processes
are not in the same path in the syntax tree
(they are both leaves by construction)
but the automaton can only examine one path in $\dataset$ at a time;
we then first guess the sender, mark the channel carrying its message,
then select a receiver waiting on that channel
(which will be in the path of both processes)
and then spawn their continuations in the relevant places.
This requires separate steps and could lead to spurious deadlocks
when no process is listening over the selected channel.
These deadlocked states can be pruned from the bisimulation by restricting the
relevant transition system to those configurations where the control state
is a distinguished state that signals that
the intermediate steps of a synchronisation have been completed.
A successful round follows very closely the operations used
in the proof of Theorem~\ref{th:typed-tshaped}.

A round starts from a configuration with control state $\qready$,
then goes trough a number of intermediate steps until it either deadlocks or
reaches another configuration with control state $\qready$.
Only reachable configurations of $\AutEnc{P}$ with $\qready$ as control state
will correspond to reachable terms of $P$.
Thus, given an automaton
$\Aut = (\States, \delta, \qready, f_0)$,
we define the transition relation
  $(\readyto) \subseteq \CMF^2$
as the minimal relation such that $f \readyto f'$ if
$(\qready, f) \autto (q_1,f_1) \autto \cdots \autto (q_n,f_n) \autto (\qready, f')$
where in the possibly empty sequence of $(q_i, f_i)$, $q_i \neq \qready$.

To encode a reachable term $Q$ in a configuration $(\qready, f)$
we use $f$ to represent the forest $\Phi(Q)$:
roughly speaking we represent a node $n$ of $\Phi(Q)$ labelled with $l$
with a data value $d$ mapped to a $q_l$ by $f$.
Since in general, due to the generation of unboundedly many names,
there might be infinitely many such labels $l$
we need to show that we can indeed use only a finite number of distinct labels
to be able to represent them with control states.
This is achieved by using the concept of derivatives.
The set of \emph{derivatives} of a term $P$ is
the set of sequential subterms of $P$, both active or not active.
\ifshortversion
More formally, it is the set defined by the function
  $\deriv(\zero) \is \set{} $,
  $\deriv(\new x.P) \is \deriv(P) $,
  $\deriv(P \parallel Q) \is \deriv(P) \union \deriv(Q) $,
  $\deriv(\bang{M}) \is \set{\bang{M}} \union \deriv(M)$,
  $\deriv(M+M') \is \set{M+M'} \union \deriv(M) \union \deriv(M') $,
  $\deriv(\pi.P) \is \set{\pi.P} \union \deriv(P)$.
\else
More formally, it is the set defined by the following function
\begin{align*}
  \deriv(\zero) &\is \set{} \\
  \deriv(\new x.P) &\is \deriv(P) \\
  \deriv(P \parallel Q) &\is \deriv(P) \union \deriv(Q) \\
  \deriv(\bang{M}) &\is \set{\bang{M}} \union \deriv(M)\\
  \deriv(M+M') &\is \set{M+M'} \union \deriv(M) \union \deriv(M') \\
  \deriv(\pi.P) &\is \set{\pi.P} \union \deriv(P)
\end{align*}
\fi
Clearly, $\deriv(P)$ is a finite set.
Every active sequential subterm of a term $P'$ reachable from $P$ is congruent
to a $Q\sigma$ for some substitution $\sigma$.
When $P$ is depth-bounded, we know from \cite{meyer:db} that,
there is a finite set of substitutions such that
the substitution $\sigma$ above can always be drawn from this set.
The assumption that $P$ is \tshaped\ and typable allows us to be even more specific.
Let $X_\Types = \set{\chi_t | t \in \Types}$ be a finite set of names, we define
$\Deriv[P] \is
  \set{ Q \sigma |
          Q \in \deriv(P),
          \sigma \from \freenames(Q) \to (X_\Types \union \freenames(P))
  }$.

\begin{lemma}
\label{lemma:chi-type}
Let $P$ be a term such that $\forest(P)$ is \tcompat.
Then there exists a term $Q$ such that $\forest(Q)$ is \tcompat,
$Q$ is an \pre\alpha-renaming of $P$,
$\resboundnames(Q) \subseteq X_\Types$ and
each active sequential subterm of $Q$ is in $\Deriv[P]$.
\end{lemma}

\begin{proof}
By definition of \tcompat\ forest we have that
in any path of $\forest(P)$ no two distinct nodes will have labels
$(x, t)$ $(x', t)$ so \pre\alpha-renaming
each restriction $(x \tas \type)$ of $P'$ to $(\chi_{\base(\type)} \tas \type)$
will yield the desired $Q$.
\end{proof}

Henceforth, we will write $\Phi'(P)$
for a relabelling of the forest $\Phi(P)$ such that its
labels use only names in $X_\Types$, as justified by Lemma~\ref{lemma:chi-type}.

\begin{corollary}
If a term $P$ is typably hierarchical, then every $P' \in \reach(P)$
is congruent to a term $Q$ such that
$\resboundnames(Q) \subseteq X_\Types$ and each active sequential subterm of $Q$
is in~$\Deriv[P]$.
\end{corollary}
\begin{proof}
By Theorem~\ref{th:typed-tshaped} and Lemma~\ref{lemma:chi-type}.
\end{proof}

The transition relation of the automaton encoding of a term $P$
is then derived from the set $\Deriv$.

Before we show how to construct the transitions of the automaton from the term,
we define a relation $\bisimA$ between terms and class memory functions.
This relation formalises how we encode the term as a labelling of data values,
and will have a crucial role in proving the soundness of the encoding.
\input{definitions/bisimA}

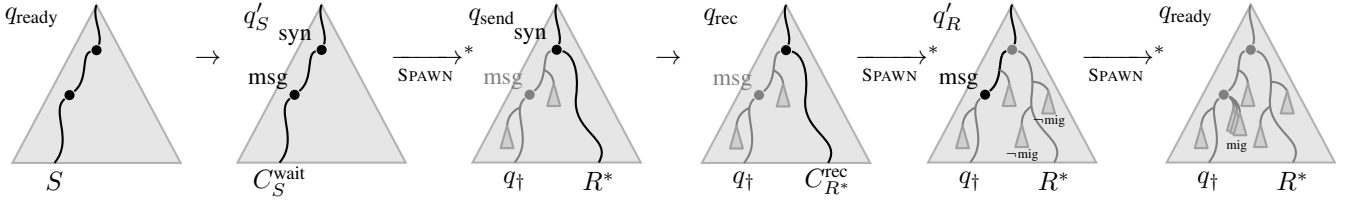
\begin{figure*}
  \ifshortversion%
    \tikzset{baseline=-1cm, x=.43cm, y=.5cm}%
  \else%
    \tikzset{baseline=-1cm, x=.43cm, y=.6cm}%
  \fi%
  \input{figures/encoding-explain-notikz}
  \ifshortversion\unskip\fi
  \caption{%
    A schema of the transitions simulating a synchronisation
    in the automaton encoding of a term.
    The trees represent the class memory functions associated with
    configurations in the run of the automaton.
    The run simulates a sender synchronising with a replicated receiver.
    The two displayed nodes in the path leading to $S$ are the ones
    labelled with the names of, from top to bottom,
      the synchronisation channel and
      the exchanged message.%
  }
  \label{fig:encoding-explain}
\end{figure*}

Let us now describe how we can simulate reduction steps of a \piterm\ with
transitions in a \NDA.
In encoding a \piterm's semantics into the transition relation of a \NDA,
we need to overcome the differences in the primitive steps allowed in the two models.
Simulating a \picalc\ synchronisation requires matching two paths,
leading to the two reacting sequential terms,
in $\dataset$ at the same time.
A step in the automata semantics can only manipulate a single path,
so we will need to split the detection of a redex in two phases:
finding the sender, then finding a matching receiver.
Moreover, finding a redex requires detecting that the path under consideration
contains a node labelled with the synchronising channel
and one with the appropriate sequential term,
ignoring how many and which other nodes are in between them.
To succinctly represent this operation, we introduce the following notation.
Fix a set $\States$ including $q, q', \lstc{l}{n}, \lstc{l'}{n}, l$.
We associate to the expression
  $\enctr{q,\lst{l}{n}}{q',\lst{l'}{n}}$
the set of transitions
\begin{multline*}
  \tran_{\States}(\enctr{q,\lst{l}{n}}{q',\lst{l'}{n}}) \is\\
  \bigl\{
    (q, \lstc{q}{m}, q', \lstc{q'}{m}) \in \States^{2m+2}
    \mid
      \exists\, \lst{i}{m}\st
    \\
        1 \leq i_1 < \dots < i_n \leq m,
        q_{i_j} = l_j, q'_{i_j} = l'_j
  \bigr\}.
\end{multline*}
When the sequence $\lstc{l}{n}$ is empty, the expression simply means
that the automaton may go from a configuration $(q,f)$ to $(q',f)$
with no condition (nor effect) on $f$.
Similarly, we associate to the expression
  $\enctradd{q,\lst{l}{n}}{q',\lst{l'}{n}}{l}$
the set of transitions
\begin{multline*}
  \tran_{\States}(\enctradd{q,\lst{l}{n}}{q',\lst{l'}{n}}{l}) \is\\
  \bigl\{
    (q, \lstc{q}{m}, \fresh, q', \lstc{q'}{m}, l) \in \States^{2m+4}
    \mid
      \exists\, \lst{i}{m}\st
    \\
        1 \leq i_1 < \dots < i_n = m,
        q_{i_j} = l_j, q'_{i_j} = l'_j
  \bigr\}.
\end{multline*}
Note that the sequence $\lstc{l}{n}$ may be empty,
in which case the data value labelled with $\fresh$
is selected among the level-1 ones.
The set of states mentioned in an expression is
$\enctrstates(\enctr{q,\lst{l}{n}}{q',\lst{l'}{n}}) \is \set{q,q',\lstc{l}{n},\lstc{l'}{n}}$
  and
$\enctrstates(\enctradd{q,\lst{l}{n}}{q',\lst{l'}{n}}{l}) \is \set{q,q',l,\lstc{l}{n},\lstc{l'}{n}}$.

To define the transitions of the encoding of a term,
we make use of some auxiliary definitions generating sets of transition expressions.
\ifshortversion%
  For lack of space, their formal definition can be found in~\appendixorfull.
\fi

$\trgen{Setup}(q,q',l,\phi)$ adds to the path leading to a data value
labelled with $l$, the nodes corresponding to a forest $\phi \in \AST{Q}$
for some $Q$.
These transitions are deterministic in the sense that a configuration
$(q, f)$ with only one data value labelled with $l$ will transition through all
the transitions dictated by $\trgen{Setup}(q,q',l,l',\phi)$ reaching $(q', f')$.
\iflongversion%
  \input{definitions/setup}
\fi

Similarly, we define $\trgen{Spawn}(q,q',l,l',\phi)$ to be
the set of transitions needed to append each tree in $\phi$ to nodes
in the path leading to a data value $d$ labelled with $l$;
the operation starts at control state $q$ and ends at control state $q'$
with the label for $d$ updated to $l'$.
Each tree is appended to the node with the lowest level such that every
name mentioned in its leaves is an ancestor of such node.
Since a single transition can add only one node of $\phi$,
we need a  number of transitions to complete the operation;
these transitions will however be deterministic in the same sense as the ones
required to complete a $\trgen{Setup}$ operation.
\iflongversion%
  \input{definitions/spawn}
\fi

\ifshortversion
Next, we define for each $D \in \Deriv$ the set of transition expressions
$\trgen{React}(D)$ representing the steps needed to simulate in the automaton
the potential reactions of $D$:
  $\trgen{React}(M) \is \trgen{React}_{\qdead}^M(M)$,
  $\trgen{React}(\bang{M}) \is \trgen{React}_{\bang{M}}^{\bang{M}}(M)$.
The set of transition expressions $\trgen{React}_q^D(M)$ collects all the
potential reactions of $M$ as a choice of $D$;
the label $q$ is the one that should be associated with
the ``consumed'' term $D$ after a reaction has been completed.
The transitions simulating a replicated component will not mark,
as the ones for non replicated terms, the reacted term with $\qdead$,
which will represent ``garbage'' inert nodes in $f$.
For lack of space we do not present the formal definition of $\trgen{React}$
which can be found in~\appendixorfull.
\else
  \input{definitions/react}
\fi
Figure~\ref{fig:encoding-explain} illustrates the steps the automaton
performs when simulating a synchronisation.

\begin{definition}[Automaton encoding]
  The automaton encoding of a typably hierarchical term $P$
  is the \NDA\ $\AutEnc{P} = (\States, \delta, \qready, f)$
  where
    $\trgen{Tr} = \Union \set{\trgen{React}(D) | D \in \Deriv[P]}$,
    $\States = \enctrstates(\trgen{Tr})$,
    $\delta = \tran_\States(\trgen{Tr})$ and
    $f$ is an arbitrary class memory function such that
    $P \bisimA f$.
\end{definition}

\subsection{Soundness of the encoding}
\label{sub:bisim}

In this section we will show that the transition system of the semantics of $P$
is bisimilar to the one of $\Aut$ when restricting it to configurations
with control state equal to $\qready$.

A transition system is a tuple $(S, \to, s)$ where $S$ is a set of configurations,
$(\to) \subseteq (S \times S)$ is the transition relation and
$s \in S$ is the initial state.
Two transition systems $(S_1, \to_1, s_1)$ and $(S_2, \to_2, s_2)$
are said to be \emph{bisimilar} if there exists a relation
$(\bisim) \subseteq S_1 \times S_2$
such that $s_1 \bisim s_2$ and $\bisim$ is a \emph{bisimulation},
that is, if $s \bisim t$ then:
\begin{enumerate*}[label=(\Alph*)]
  \item for each $s' \in S_1$ such that $s \to_1 s'$
        there is a $t' \in S_2$ such that $t \to_2 t'$ and $s' \bisim t'$;
  \item for each $t' \in S_2$ such that $t \to_2 t'$
        there is a $s' \in S_1$ such that $s \to_1 s'$ and $s' \bisim t'$.
\end{enumerate*}
Establishing that two transition systems are bisimilar implies that
a wide class of properties are preserved across bisimilar states.
For our purposes, proving that the automaton encoding of a term
gives rise to a bisimilar transition system has the important consequence
that reachability can be reduced from one model to the other.

\begin{theorem}
\label{th:bisim-encoding}
The transition system $(\CMF, \readyto, f_0)$
induced by the automaton $\AutEnc{P} = (\States, \delta, \qready, f_0)$
obtained from a closed typably hierarchical term $P$,
is bisimilar to the transition system of the reduction semantics of $P$,
$(\reach(P), \to, P)$.
\end{theorem}

The result is proved by showing that the relation $\bisimA$ defined above,
is a bisimulation that relates the initial states of the two transition systems.
\ifshortversion
  A proof sketch can be found in~\appendixorfull.
\else
  \input{proofs/encodingA}
\fi

\subsection{Encoding of \NDA\ into Typably Hierarchical terms}
\label{sec:nda2pi}

In this section we sketch how an \NDA\ can be encoded into
a bisimilar typably hierarchical \piterm.
\ifshortversion
  As this encoding is less surprising, we do not present it in full,
  the technical definitions and proofs can be found in~\appendixorfull.
\fi

Similarly as the encoding in the opposite direction,
the \picalc\ encoding of an automaton \Aut\ will
represent a reachable configuration $(q,f)$
using the forest of a reachable term $P$.
A term representing a reachable configuration may need to execute several steps
before reaching another term representing a successor configuration.

Fix an automaton $(\States, \delta, q_0, f_0)$.
For simplicity we show the case where $\forall d \st f_0(d) = \fresh$,
the general case follows the same scheme.
First we note that every transition in $\delta_i$ is of the form
\shortversioninline{
  (q_0, \lst{q}{j}, \underbrace{\fresh,\dots,\fresh}_{i-j},  q'_0, \lst{q'}{i})
}
for some $1 \leq j \leq i$, where $q_k\in \States$ for all $0\leq k\leq j$.
Instead of using the partition $\delta = \Union_{i=1}^\ell \delta_i$
we re-partition the transition relation as
$\delta = \Union_{j=0}^\ell \theta_j$
where \shortversioninline{
  \theta_j \is \Union_{i=j}^\ell \set{
    (q_0, \lst{q}{j},
      \underbrace{\fresh,\dots,\fresh}_{i-j},
     q'_0, \lst{q'}{i})
     \in \delta_i
   }
}
(fixing $\delta_0 = \emptyset$ for uniformity).
We introduce a channel name $c^i_q$
  for each $q \in \States$
  and each level of the automaton $i$.
Our encoding will show no mobility,
so each such channel $c$ will have type $t_c$,
hence no message will be exchanged on synchronisation;
we abbreviate this kind of synchronisation with $c.P$ and $\outz c.Q$.%
\footnote{%
  It is easy to see that this can be accommodated in our syntax
  by assuming a global name $r$, typed with a type $t_r$
  that is set to be the parent of each root in $\Types$;
  a synchronisation over a channel $c \tas t_c[t_r]$ without exchanging a message
  is then represented by $\inp c(x).P$ and $\out c<r>.Q$
  with $x \not\in \freenames(P)$.%
}
\newcommand{\tr}{\mathit{tr}}%
Let $\Chan^i \is \set{(c^i_q \tas t_{c^i_q}) | q \in \States }$.
Given a transition $\tr \in \theta_j$ where
$\tr = (q_0, \lst{q}{j}, \fresh,\dots,\fresh,  q'_0, \lst{q'}{i})$
we define the term $A_\tr$ to be
\[
  A_\tr \is
    c^0_{q_0}.\cdots.c^j_{q_j}.
      \new \Chan^{j+1}. \cdots \new \Chan^{i}.
        \left(
        \Parallel*_{k = 0}^i \outz c^k_{q'_k}
          \parallel
        \Parallel*_{k = j+1}^i P_{\theta_{k}}
        \right)
\]
where $P_{\theta_j} \is \Parallel_{\tr \in  \theta_j} \bang{(A_\tr)}$
and $\Parallel*_{k = i+1}^i P_{\theta_{k}} = \zero$.
Note that these definitions are well-defined since they are not recursive.
The \emph{\piterm\ encoding} of the \NDA\ $\Aut = (\States, \delta, q_0, f_0)$ is then defined as
$\ProcEnc{\Aut} \is \new \Chan^0.( P_{\theta_0} \parallel \outz c^0_{q_0} )$.

Similarly to our previous result, the encoding needs more than one step
to simulate a single transition of the automaton.
Hence, to state the result on the correspondence
between the semantics of the automaton and its encoding,
we define a derived transition system on \piterm{s} as follows.
Let $P$ and $Q$ be two \piterm{s} such that $P \to^+ Q$,
if  $P \congr \new \Chan^0.(\outz c \parallel P')$
and $Q \congr \new \Chan^0.(\outz c' \parallel Q')$
with $c$, $c' \in \Chan^0$,
and none of the intermediate processes in the reduction from $P$ to $Q$
is in that form,
then $P \chanzto Q$.
Note that even after \pre\alpha-renaming a term in the encoding,
we would be able to pinpoint names from each $\Chan^i$
by looking at their types, as \pre\alpha-renaming does not affect type annotations.

\begin{theorem}
  \label{th:bisim-encoding-bis}
  The transition system generated by
  the semantics of a level-$\ell$ \NDA\ \Aut\ and
  the transition system $\chanzto$ with $\ProcEnc{\Aut}$ as initial state,
  are bisimilar.
\end{theorem}

\iflongversion
\begin{proof}
  \input{proofs/encodingP}
\end{proof}
\fi

\begin{theorem}
  \label{th:encoding-is-hierarchical}
  $\ProcEnc{\Aut}$ is typably hierarchical.
\end{theorem}

\iflongversion
\begin{proof}
  \input{proofs/encodingP-hierarchical}
\end{proof}
\fi

%% file: definitions/bisimA.tex
Let $Q$ be a term reachable from $P$ and
$(\qready, f)$ be a configuration of an automaton $\Aut$.
Let $\phi = \Phi'(Q)$,
the relation $Q \bisimA f$ holds if and only if
there exists an injective function $\iota \from \nodes(\phi) \to \dataset$
such that for all $n \in \nodes(\phi)$:
\begin{enumerate}[label=\roman*)]
  \item if $\iota(n) = d$, $n' \parent_\phi n$ and $\iota(n') = d'$ then
        $d' = \pred(d)$;
  \item if $n$ is labelled with $(\chi_i, t)$ then $f(\iota(n)) = \chi_i$;
  \item if $n$ is labelled with a sequential process $Q'$ then
        $f(\iota(n)) = Q'$;
  \item for each $d$ such that $f(d) \neq \fresh$ either
          there is an $n$ such that $\iota(n) = d$
            or
          $f(d) = \qdead$.
\end{enumerate}

%% file: figures/encoding-explain-notikz.tex
\begingroup
\newcommand{\afterastep}{\raisebox{1.8cm}{\makebox[2em]{$\to$}}}%
\newcommand{\aftersomesteps}[1]{\raisebox{1.8cm}{\makebox[2em]{$\xrightarrow[#1]{}^*$}}}%
\includegraphics{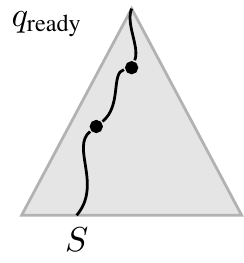}%
\afterastep%
\raisebox{-2.1pt}{\includegraphics{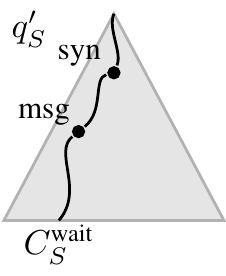}}%
\aftersomesteps{\trgen{Spawn}}%
\includegraphics{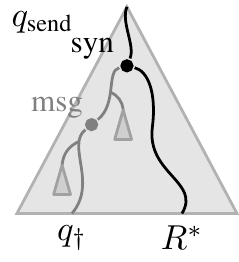}%
\afterastep%
\raisebox{-1.4pt}{\includegraphics{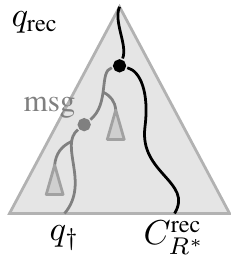}}%
\aftersomesteps{\trgen{Spawn}}%
\includegraphics{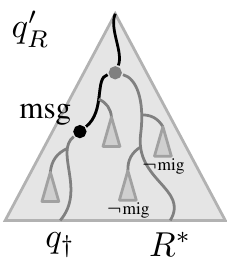}%
\aftersomesteps{\trgen{Spawn}}%
\includegraphics{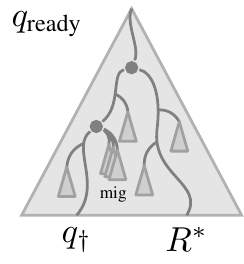}%
\endgroup

%% file: definitions/setup.tex
Formally, suppose, for some $j$ and $k$,
  $\phi = \set{(x_1, \type_1)[\phi_1],\dots,(x_j, \type_j)[\phi_j]}
            \union
          \set{Q_1[],\dots,Q_k[]}$
where
  all $x_i$ are in $X_\Types$ and
  all $Q_i \in \Deriv$.
Then $\trgen{Setup}$ is defined as follows:
\bgroup
\newcommand{\nudge}{\\&\hspace{-4.5em}\union{}} 
\begin{align*}
  \trgen{Setup}(q,q',l,l',\phi)
    &  \is \set{\enctradd{q, l}{q_1, l}{\s{Q_1}{ready}}}
    \nudge \set{\enctradd{q_i,l}{q_{i+1},l}{\s{Q_{i+1}}{ready}} | 1 \leq i \leq j}
    \nudge \set{\enctradd{q_j,l}{q'_1,l}{\s{Q_j}{ready}}}
    \nudge \set{\enctradd{q'_i,l}{q''_i,l}{\s{x_i}{set}} | 1 \leq i \leq k}
    \nudge \Union_{i=1}^k \trgen{Setup}(q''_i,q'_{i+1},\s{x_i}{set},x_i,\phi_i)
    \nudge \set{\enctr{q'_{k+1},l}{q',l'}}
\end{align*}
where for all $1 \leq i \leq j$ and all $1 \leq i' \leq k$,
$q_i, q'_{i'}, q''_{i'}, q'_{k+1}$ are fresh intermediate control states.
in the sense that they are only mentioned in the transitions
generated by that specific application of $\trgen{Setup}$.
We allow $l$ to be the empty sequence,
in which case $l'$ needs to be the empty sequence as well.
\egroup

%% file: definitions/spawn.tex
Formally, let the forest $\phi = \Phi'(D)$ consist of trees $\lstc{\theta}{k}$,
for a term $D \in \Deriv$.
We can precompute, for each $\theta_i$, the base type
  $t_i \is \min_{\parent_\Types}\set{
    t | \chi_t \in \freenames(A), n \in N_{\theta_i}, \ell_{\theta_i}(n) = A}$
when defined.
For each label $\chi_t \in X_\Types$ we also have a label $\s{\chi_t}{sp}$
we write $\chi(\theta_i)$ (resp. $\s{\chi}{sp}(\theta_i)$)
for $\chi_{t_i}$ (resp. $\s{\chi_{t_i}}{sp}$) when $t_i$ is defined,
or the empty sequence when $t_i$ is undefined
(e.g.~when $\theta_i$ does not have free variables).
Then $\trgen{Spawn}(q,q',l,l',\phi)$ is the set of transition expressions
defined as follows:
\bgroup
\newcommand{\nudge}{\\&\hspace{-7.5em}} 
\begin{align*}
  \trgen{Spawn}(q_0,q',l,l',\phi)
    \is\nudge
    \set{\enctr{q_{i-1}, \chi(\theta_i)\:l}{q'_{i-1}, \s{\chi}{sp}(\theta_i)\:l}
          | 1 \leq i \leq k}
    \nudge\union{}
    \Union_{i=1}^k
      \trgen{Setup}(q'_{i-1}, q_i, \s{\chi}{sp}(\theta_i), \chi(\theta_i),\theta_i)
\end{align*}
where for all $1 < h \leq k$, $q_h, q'_h$ are fresh.
\egroup

%% file: definitions/react.tex
We define for each $D \in \Deriv$ the set of transition expressions
$\trgen{React}(D)$ representing the steps needed to simulate in the automaton
the potential reactions of $D$.
\begin{align*}
  \trgen{React}(M) &\is \trgen{React}_{\qdead}^M(M) \\
  \trgen{React}(\bang{M}) &\is \trgen{React}_{\bang{M}}^{\bang{M}}(M)
\end{align*}
The set of transition expressions $\trgen{React}_q^D(M)$ collects all the
potential reactions of $M$ as a choice of $D$;
the label $q$ is the one that should be associated with
the ``consumed'' term $D$ after a reaction has been completed.
The transitions simulating a replicated component will not mark,
as the ones for non replicated terms, the reacted term with $\qdead$,
which will represent ``garbage'' inert nodes in $f$.
The term $\zero$ cannot initiate any step and a choice may do any action
that one of its choices can:
\begin{align*}
  \trgen{React}_q^D(\zero) &\is \emptyset \\
  \trgen{React}_q^D(M+M') &\is \trgen{React}_q^D(M) \union \trgen{React}_q^D(M')
\end{align*}
Any sender can initiate a synchronisation from the ready state:
\bgroup
\newcommand{\nudge}{\\&\hspace{-7.5em}} 
\begin{align*}
  \trgen{React}_q^D(\out \chi_t<\chi_{t'}>.C) \is
    \nudge
    \set{\enctr{\qready,  \chi_t \chi_{t'} D}
               {q', \s{\chi_t}{syn} \s{\chi_{t'}}{msg} \s{C}{wait}}
      | t < t'}
    \nudge\union
    \set{\enctr{\qready,  \chi_{t'} \chi_t D}
               {q', \s{\chi_{t'}}{msg} \s{\chi_t}{syn} \s{C}{wait}}
      | t > t'}
    \nudge\union
    \trgen{Spawn}(q', \q{send}, \s{C}{wait}, q, \Phi'(C)).
\end{align*}
\egroup
where $q'$ is fresh.
Here, the state $\q{send}$ signals that we are in the middle of a synchronisation,
where the sender is committed but a receiver has yet to be selected.

For the case of an input prefix $M = \inp\chi_t(x).C$ we distinguish two cases:
when the base type of $\chi_t$ is greater than the base type of $x$
no migration occurs, otherwise part of the continuation needs to be spawned
in the sender's path.
In the case when the base type of $\chi_t$ is greater than the base type of $x$,
we set
\bgroup
\newcommand{\nudge}{\\&\hspace{-7.5em}} 
\begin{align*}
  \trgen{React}_q^D(\inp\chi_t(x).C) \is
    \nudge
    \set{\enctr{\q{send}, \s{\chi_t}{syn} \s{\chi_{t'}}{msg} D}
               {\q{rec},  \chi_t \chi_{t'} \s{C}{rec}}
        | t \tlt t' \in \Types}
    \nudge
    \set{\enctr{\q{send}, \s{\chi_{t'}}{msg} \s{\chi_t}{syn} D}
               {\q{rec},  \chi_{t'} \chi_t \s{C}{rec}}
        | t > t' \in \Types}
    \nudge\union
    \trgen{Spawn}(\q{rec}, \qready, \s{C}{rec}, q, \Phi'(C)).
\end{align*}
\egroup
In the case when the base type of $\chi_t$ is greater than the base type of $x$,
more transitions are required.
First, we precompute for each $M = \inp\chi_t(x).C$ as above
and $t \tlt t' \in \Types$,
the two forests $\mig{C, t'}$ and $\nmig{C}$ such that
$\Phi'(C\subst{x -> \chi_{t'}}) = \mig{C, t'} \dunion \nmig{C}$ and
$\mig{C, t'}$ contains all the nodes labelled with sequential terms
tied to $\chi_{t'}$ in $C\subst{x -> \chi_{t'}}$.
As we have shown in the proof of Theorem~\ref{th:typed-tshaped},
by virtue of Lemma~\ref{lemma:tied-tree},
$\mig{C, t'}$ and $\nmig{C}$ are indeed disjoint.
Then we set:
\bgroup
\newcommand{\nudge}{\\&\hspace{-7.5em}} 
\begin{align*}
  \trgen{React}_q^D(\inp\chi_t(x).C) \is
    \nudge
    \set{\enctr{\q{send}, \s{\chi_t}{syn} D}
               {\q{rec},  \chi_t \s{C}{rec}}}
    \nudge\union
    \trgen{Spawn}(\q{rec}, q',\s{C}{rec},q, \nmig{C})
    \nudge\union
    \Union_{t' \in \Types} \trgen{Setup}(
      q', \qready,
      \s{\chi_{t'}}{msg}, \chi_{t'},
      \mig{C, t'}
    )
\end{align*}
where $q'$ is a fresh intermediate control state.
\egroup

%% file: proofs/encodingA.tex
By definition of $\AutEnc{P}$ we have $P \bisimA f_0$.
Showing that $\bisimA$ is indeed a bisimulation amounts to showing that
if $Q \bisimA f$ then:
\begin{enumerate}[label=(\Alph*)]
  \item for each $Q'$ such that $Q \redto Q'$
        there is a $f'$ such that $f \readyto f'$ and $Q' \bisimA f'$;
        \label{eq:bisim-proc}
  \item for each $f'$ such that $f \readyto f'$
        there is a $Q'$ such that $Q \redto Q'$ and $Q' \bisimA f'$.
        \label{eq:bisim-aut}
\end{enumerate}
To show this holds we rely on the hypothesis that $Q \bisimA f$
to get a $\iota$ relating $\Phi'(Q)$ and $f$.
The proof then closely follows the constructions
in the proof of Theorem~\ref{th:typed-tshaped}.
If $Q \redto Q'$ we can find two nodes $n_S$ and $n_R$ in $\Phi'(Q)$
labelled with the sender and receiver processes responsible for the reduction;
they will share an ancestor $n_a$ labelled $(\chi_t, t)$
corresponding to the channel on which they are synchronising.
On the automaton side, we have that $(\qready, f)$ matches the rule generated
from the sender by selecting the data value $d_S = \iota(n_S)$,
a data value $d_b$ corresponding to the name being sent and $d_a = \iota(n_a)$.
This leads to
$(q', f)$ where
  $f'(d_a) = \s{\chi_t}{syn}$,
  $f'(d_b) = \s{\chi_{t'}}{msg}$,
  $f'(d_S) = \s{S'}{wait}$.
From here only one of the transitions generated from $\trgen{Spawn}$
of the continuation is enabled as there is only one node marked with `wait'.
The transitions are deterministic from here until a configuration
$(\q{send}, f')$ is reached with $f'$ representing the initial forest with
the continuation of the sender added and with the node of the sender updated
with either $\qdead$ or the sender itself if it is a replicated component.
At this point there is only one data value marked with `syn' and
the only transitions from $\q{send}$ are the ones generated from a process
that can receive from the marked channel.
We can pick the rule that has been generated
from the receiver involved in the reduction from $Q$ to $Q'$
and go to a configuration with control state $\q{rec}$.
From this configuration the transitions are deterministic.
The next configuration reached with control state
$\qready$ is bisimilar to $Q'$ by tracing the effects these transitions have
on the class memory function.
Fresh data values get assigned labels compatible with the non migrating
continuations of the receiver first,
and then the migrating ones as children of $d_b$;
data values with meaningless labels get assigned the label $\qdead$.

To prove~\ref{eq:bisim-aut} we proceed similarly.
Every reduction sequence from $(\qready, f)$ to $(\qready, f')$
must start with a transition to a configuration with control state $\q{send}$,
which is generated by rules extracted from a sender $S$
labelling a data value $d_S$;
since $Q \bisimA f$ we know that $n_S = \inv\iota(d_S)$ is labelled with $S$
in $\Phi'(Q)$, hence $S$ is an active sequential process of $Q$.
To complete this part of the proof we only need to follow the transitions
of the automaton in the same way as done for the previous point,
and note that the only way the automaton can reach a configuration
with control state $\qready$ from $(\qready, f)$
is by selecting a receiver that can synchronise with the selected sender.
This is important because there may be transitions from $(\qready, f)$
corresponding to selecting a sender trying to synchronise
on a channel on which no receiver is listening.
This transition would lead to a deadlocked configuration
(one with no successors)
but never going through a configuration with control state $\qready$.

%% file: proofs/encodingP.tex
Fix an \NDA\ $\Aut = (\States, \delta, q_0, f_0)$
with $\delta = \Union_{0 \leq j \leq \ell} \theta_j$ as before.
We prove the theorem by exhibiting a bisimulation relation
$(\bisimP) \subseteq
  (\States \times (\dataset \pto \States_\fresh)) \times \reach(\ProcEnc{\Aut})$
between the two transition systems.
For a class memory function $f \from \dataset \to \States_\fresh$,
let $f(\dataset)$ be the $\States$-labelled forest with
the set $N=\set{ d \in \dataset | f(d) \neq \fresh }$ as nodes,
each labelled with $f(d)$
and with $\pred_\dataset$ restricted to $N$ as parent relation.
We first define a hierarchy of relations $\bisimP_i$
between $\States$-labelled forests and \piterm{s},
for $0 \leq i \leq \ell$, as follows:
$
  q[\set{\lstc{\phi}{n}}]
    \bisimP_i
  \new \Chan^i.(
    P_{\theta_i} \parallel
    \outz c^i_q \parallel
    \Parallel_{1 \leq j \leq n} P_j
  )
$
if, for all $1 \leq j \leq n$, $\phi_j \bisimP_{i+1} P_j$.
Since $n$ must be 0 for $i=\ell$, the relation is well-defined.
Let $P \in \reach(\ProcEnc{\Aut})$ and $(q, f)$ be a reachable configuration of \Aut.
Then $(q, f) \bisimP P$ if there exists a $P' \congr P$ such that
$q_0[f(\dataset)] \bisimP_0 P'$.
To show that $\bisimP$ is indeed a bisimulation, we have to prove that
if $(q,f) \bisimP P$ then:
\begin{enumerate}[label=(\Alph*)]
  \item for each $(q',f')$ such that $(q,f) \autto (q',f')$
        there is a $P'$ such that $P \chanzto P'$ and $(q',f') \bisimP P'$;
        \label{eq:bisim-bis-aut}
  \item for each $P'$ such that $P \chanzto P'$
        there is a $(q',f')$ such that $(q,f) \autto (q',f')$
        and $P' \bisimP (q',f')$.
        \label{eq:bisim-bis-proc}
\end{enumerate}
To prove \ref{eq:bisim-bis-aut} we proceed as follows;
suppose $(q,f) \autto (q',f')$ is an application of a transition
$t = (q, \lst[1]{q}{j}, \fresh,\dots,\fresh,
  q', \lst[1]{q'}{i}) \in \theta_j$
then the forest $q[f(\dataset)]$ has a path from the root to a leaf
labelled with $q,\lstc{q}{j}$, which, by definition of $\bisimP$,
implies that $P$ is congruent to a term with the following shape:
\[
  \new \Chan^0.(
    R_0 \parallel
    \outz c^0_q \parallel
    \new \Chan^1.(
      R_1 \parallel
      \outz c^1_{q_1} \parallel
      \cdots
        \new \Chan^j.(
          R_j \parallel
          \outz c^j_{q_j} \parallel
          P_{\theta_j}
        )
      \cdots
    ).
\]
By construction, $P_{\theta_j} \congr \bang{(A_\tr)} \parallel R$ and
$A_\tr$ is a process inputting once from $c^0_q$ then
once from each $c^k_{q_k}$ in sequence.
From the shape of $P$ we can conclude all of these input prefixes
can synchronise with the dual $\outz c^k_{q_k}$ processes
in parallel with them, activating, in $j+1$ steps, the continuation
$C =
  \new \Chan^{j+1}. \cdots \new \Chan^{i}.
    \left(
    \outz c^0_{q'} \parallel
    \Parallel_{k=1}^\ell \outz c^k_{q'_k}
      \parallel
    P_{\theta_{j+1}}
    \right)$,
yielding the process
\[
  P' \congr \new \Chan^0.(
    R_0 \parallel
    \outz c^0_{q'} \parallel
    \new \Chan^1.(
      R_1 \parallel
      \outz c^1_{q'_1} \parallel
      \cdots
        \new \Chan^i.(
          R_i \parallel
          \outz c^j_{q'_j}
        )
      \cdots
    )
\]
where for $k$ between $j+1$ and $i$, $R_k = P_{\theta_k}$.
Now consider the forest $q'[f'(\dataset)]$:
it coincides with $q[f(\dataset)]$ except
on the path we singled out, now labelled with
$q', q'_0, \dots, q'_j$ and continuing to a leaf
with nodes labelled $q'_{j+1},\dots,q'_i$.
It is easy to see that $q'[f'(\dataset)] \bisimP P'$.

To prove \ref{eq:bisim-bis-proc} one can proceed similarly,
by observing that even if $\ProcEnc{\Aut}$ can perform some reductions which deadlock
that do not correspond to reductions of the automaton,
these steps cannot lead to a state with $\outz c^0_{q'}$
as one of the active sequential processes.
This claim is supported by the following easy to verify invariant:
in any term $P$ reachable from $\ProcEnc{\Aut}$, for each bound name $c$ in $P$
there is at most one active sequential subterm of $P$ outputting on $c$.
This is satisfied by $\ProcEnc{\Aut}$ and preserved by reduction.

%% file: proofs/encodingP-hierarchical.tex
Assume an arbitrary strict total order $<_\States$
on the automaton's control states;
let then $(\Types, \parent)$ be the forest with nodes
  $\Types = \set{t_{c^i_q} | 0 \leq i \leq \ell, q \in \States}$
and $t_{c^i_q} \parent t_{c^i_{q'}}$ if $q <_\States q'$,
and $t_{c^i_q} \parent t_{c^{i+1}_{q'}}$ if $q$ and $q'$ are respectively
the maximum and minimum states with respect to $<_\States$.
It can be proved that $\emptyset \types \nf(\ProcEnc{\Aut})$:
since no messages are exchanged over channels, the constraints on types
are trivially satisfied;
for the same reason, no sequential term under an input prefix is migratable,
making all the base type constraints in rule~\ref{rule:in} trivially valid.
The base type inequalities of rule~\ref{rule:conf} are also satisfied since
in $A_\tr$ for $\tr \in \theta_j$, every $P_{\theta_k}$ might be tied to any channel $c$ in $\Chan^{j+1}\union\dots\union\Chan^{i}$
but can only have as free names channels in $\Chan^h$ with $h \leq j$,
which all have base types smaller than $c$.

%% file: 6-discussion.tex
\section{Related Work}
\label{sec:relwork}

\emph{Depth boundedness} in the \picalc{} was first proposed in~\cite{Meyer2009c} and later studied in~\cite{meyer:db} where it is proved that depth-bounded systems are {well-structured transition systems}.
In~\cite{wies:fwd} it is further proved that (forward) coverability is decidable even when the depth bound $k$ is not known \emph{a priori}.
In~\cite{widening} an approximate algorithm for computing the \emph{cover set}---an over-approximation of the set of reachable terms---of a system of depth bounded by $k$ is presented.
All these analyses rely on the assumption of depth-boundedness and may even require a known bound on the depth to terminate.

Several other interesting fragments of the $\pi$-calculus have been proposed in the literature, such as
  name bounded~\cite{meyer:name},
  mixed bounded~\cite{meyer:mixed}, and
  structurally stationary~\cite{Meyer2009c}.
Typically defined by a non-trivial condition on the set of reachable terms -- a \emph{semantic} property, membership becomes undecidable.
Links with Petri nets via encodings of proper subsets of depth-bounded systems have been explored in \cite{meyer:mixed}.
Our type system can prove depth-boundedness for processes that are breadth and name unbounded, and which cannot be simulated by Petri nets.
Recently H\"{u}chting et al.~\cite{boundsmobility} proved several relative classification results between fragments of \picalc.
Using 
Karp-Miller trees,
they presented an algorithm to decide if an arbitrary \piterm\ is
bounded in depth by a given $k$.
The construction is based on an (accelerated) exploration of the state space of the \piterm\, which can be computationally expensive.
By contrast, our type system uses a very different technique
leading to a quicker algorithm, at the expense of precision.
Our forest-structured types can also act as specifications,
offering more intensional information to the user than just a bound $k$.

Our type system is based on Milner's sorts for the \picalc~\cite{Milner:1993},
later refined into I/O types~\cite{PierceS93} and their variants~\cite{PierceS00}.
Based on these types is a system for termination of \piterm{s}~\cite{piterm} that uses a notion of levels, enabling the definition of a lexicographical ordering.
Our type system can also be used to determine termination of \piterm{s} in an approximate but conservative way, by
composing it with a procedure for deciding termination of depth-bounded systems.
Because the respective orderings between types of the two approaches are different in conception,
we expect the terminating fragments isolated by the respective systems to be incomparable.

A rather different approach to typing \piterm{s}
is presented in~\cite{behavioural} where behavioural types are introduced.
Roughly speaking, the type system can extract from a \piterm\ $P$
a type which is itself a CCS term simulating $P$.
Properties of the type (such as absence of locks)
can then be transferred back to $P$ by virtue of this simulation.
By contrast, our types do not carry information about
the evolution of the system;
if a system is proved depth-bounded by the type system,
its evolution can be analysed quite accurately
using the decision procedures for depth-bounded systems.

{Nested Data Class Memory Automata} were introduced~\cite{conrad:14} as an extension of Class Memory Automata to operate over tree-structured datasets. 
Without the local acceptance condition, NDCMA have decidable emptiness, and in the deterministic case are closed under all Boolean operations (see~\cite{conrad:14}).
Thanks to these algorithmic properties, NDCMA have recently found applications in algorithmic game semantics \cite{conrad:15}.


Automata that support name reasoning have been used to model the $\pi$-calculus, going back to the pioneering work of History-Dependent Automata \cite{MontanariP97}.
More recently, Tzevelekos~\cite{Tzevelekos11} introduced Fresh-Register Automata (FRA), which operate on an infinite alphabet of names and use a finite number of registers to process fresh names;
crucially it can compare incoming names with previously stored ones.
He showed that \emph{finitary \piterm{s}} (i.e.~processes that do not grow unboundedly in parallelism) are finitely representable in FRA.

\section{Future Directions}
\label{sec:future}

The type system we presented in Section~\ref{sec:typesys} is very conservative:
the use of simple types, for example, renders the analysis context-insensitive.
Although we have kept the system simple so as to focus on the novel aspects, a number of improvements are possible.
First, the extension to the polyadic case is straightforward.
Second, the type system can be made more precise by using subtyping and polymorphism to refine the analysis of control and data flow.
Third, the typing rule for replication introduces a very heavy approximation:
when typing a subterm, we have no information about which other parts of the term (crucially, which restrictions) may be replicated.
\ifshortversion
  In~\appendixorfull\ we briefly sketch a possible enhancement that is sensitive to replication.
\else

  Let us explain the issue through an example.
  \input{examples/repl-rule-extension}
\fi
The formalisation and validation of this extension is a topic of ongoing research.

%% file: examples/repl-rule-extension.tex
Let
  $ A= \tact.\new b.\tact.\new c.\bang{(\out a<c> + \inp a(x).\out b<x>)} $
and consider the two terms
$P_1 = \new a.A$ and $P_2 = \new a.\bang{A}$.
The typing derivations for the two terms are almost identical and the set of constraints they impose on $\Types$ is the same.
However $P_1$ is depth bounded, $P_2$ is not.
Therefore the type system must reject both.
We briefly sketch a possible enhancement that is sensitive to replication.
Take the term
  $ \new b\tas t_b[t]. \new l\tas t_l[t]. \new r\tas t_r[t].
      \bang{\inp b(x).\inp l(y).(\out r<x> \parallel \out b<x>)} $
which acts as a 1 cell buffer between $l$ and $r$.
This term cannot be typed by the current type system because
  $\inp l(y).(\out r<x> \parallel \out b<x>)$
is migratable for the input $\inp b(x)$
thus requiring $t_l \tlt t_b$,
but at the same time $\out b<y>$ is migratable for $\inp l(y)$
requiring $t_b \tlt t_l$, leading to contradiction.
We propose to add to the structure of $\Types$ a notion of multiplicities
of base types; a base type can be marked with either $1$ or $\omega$.
Suppose the forest of a term has a path $p$ from a node $n$ to a node $n'$
where the trace of $p$ consists only of base types marked with $1$.
This situation will represent the fact that no branching will ever occur between the two replications corresponding to $n$ and $n'$ and having one of the two names in the scope guarantees that the other one is in the scope too.
In other words, all the restrictions represented by nodes in $p$ can be though as a indivisible unit;
when typing an input term on a name with base type $t$,
the constraints of rule~\ref{rule:in} can be relaxed to require the free variables of migratable terms to have base types smaller
than the lowest $t'$ such that the path between $t$ and $t'$ in $\Types$ is formed only of base types with multiplicity $1$.
In the case of buffer example, we observe that $b$, $l$ and $r$ could all be assigned base types of multiplicity $1$ thus replacing the two conflicting constraints with the constraints $t_l \leq t'$ and $t_b \leq t'$ where $t'$ is the greatest among $t_l$, $t_r$ and $t_b$.

%% file: appendix/2-picalc.tex
\section{Supplementary Materials for Section~\ref{sec:prelim}}
\subsection{Definition of $\nf$}

\begin{figure*}[tb]
  \centering
  \input{definitions/nf}
  \caption{Definition of the $\nf \from \PiTerms \to \PiNf$ function.}
  \label{fig:nf}
\end{figure*}

The function $\nf \from \PiTerms \to \PiNf$,
defined in Figure~\ref{fig:nf},
extracts, from a term, a normal form structurally equivalent to it.

\begin{lemma}
  For each $P \in \PiTerms$,
 $P \congr \nf(P)$
\end{lemma}

\begin{proof}
A straightforward induction on $P$.
\end{proof}

\subsection{Proof of Lemma~\ref{lemma:forest-nf}}

\input{proofs/forest-nf}

%% file: appendix/3-tcompat.tex
\section{Supplementary Materials for Section~\ref{sec:t-compat}}

\subsection{Proof of Lemma~\ref{lemma:tied-tree}}

\begin{proof}
  \input{proofs/tied-tree}
\end{proof}

\subsection{Proof of Lemma~\ref{lemma:tcompat-takeout}}

\begin{proof}
  \input{proofs/tcompat-takeout}
\end{proof}

\subsection{Proof of Lemma~\ref{lemma:phi-tcompat}}

\begin{proof}
  \input{proofs/phi-tcompat}
\end{proof}

%% file: appendix/4-typesys.tex
\section{Supplementary Materials for Section~\ref{sec:typesys}}

\subsection{Proof of Lemma~\ref{lemma:subst}}

\begin{proof}
   \input{proofs/substitution}
\end{proof}

\subsection{Proof of Theorem~\ref{th:typed-tshaped}}

\begin{proof}
   \input{proofs/typed-tshaped}
\end{proof}

\subsection{Role of $\phimig$, $\phinonmig$ and $\treeins$}

To illustrate the role of $\phimig$, $\phinonmig$
and the $\treeins$ operation in the above proof,
we show an example that would not be typable if we choose a simpler
``migration'' transformation.

\input{examples/phimig}

%% file: appendix/5-encodings.tex
\section{Supplementary Materials for Section~\ref{sec:automata}}

\subsection{Definition of $\trgen{Setup}$ and $\trgen{Spawn}$}

$\trgen{Setup}(q,q',l,\phi)$ adds to the path leading to a data value
labelled with $l$, the nodes corresponding to a forest $\phi \in \AST{Q}$
for some $Q$.
\input{definitions/setup}
Note that these transitions are deterministic in the sense that a configuration
$(q, f)$ with only one data value labelled with $l$ will transition through all
the transitions dictated by $\trgen{Setup}(q,q',l,l',\phi)$ reaching $(q', f')$.

Similarly, we define $\trgen{Spawn}(q,q',l,l',\phi)$ to be
the set of transitions needed to append each tree in $\phi$ to nodes
in the path leading to a data value $d$ labelled with $l$;
the operation starts at control state $q$ and ends at control state $q'$
with the label for $d$ updated to $l'$.
Each tree is appended to the node with the lowest level such that every
name mentioned in its leaves is an ancestor of such node.
Since a single transition can add only one node of $\phi$,
we need a  number of transitions to complete the operation;
these transitions will however be deterministic in the same sense as the ones
required to complete a $\trgen{Setup}$ operation.

\input{definitions/spawn}
\subsection{Definition of $\trgen{React}$}

\input{definitions/react}
\subsection{Proof sketch of Theorem~\ref{th:bisim-encoding}}

Let us recall the definition of the relation $\bisimA$.

\input{definitions/bisimA}
Theorem~\ref{th:bisim-encoding} is proved by showing that
the relation $\bisimA$ defined above, is a bisimulation
that relates the initial states of the two transition systems.

\input{proofs/encodingA}
\subsection{Proof sketch of Theorem~\ref{th:bisim-encoding-bis}}

\input{proofs/encodingP}
\subsection{Proof sketch of Theorem~\ref{th:encoding-is-hierarchical}}

We can prove that $\ProcEnc{\Aut}$ is typably hierarchical
by fixing a specific forest of base types.

\input{proofs/encodingP-hierarchical}

%% file: appendix/6-discussion.tex
\section{Supplementary Materials for Section~\ref{sec:future}}

The typing rule for replication introduces a very heavy approximation:
when typing a subterm, we have no information about which other parts of the term (crucially, which restrictions) may be replicated.
Let us explain the issue through an example.
\input{examples/repl-rule-extension}
The formalisation and validation of this extension is the topic of ongoing research.